\documentclass[11pt]{article}
\usepackage[margin=1in]{geometry}

\usepackage[comma]{natbib}
\usepackage{url}
\usepackage{amsmath,amssymb,amsthm}
\usepackage{graphicx}

\title{Elements of Conformal Prediction}

\author{
Matteo Sesia\thanks{Department of Data Sciences and Operations, and Thomas Lord Department of Computer Science, University of Southern California, Los Angeles, CA 90089, USA. Email: \texttt{sesia@marshall.usc.edu}}
\and
Stefano Favaro\thanks{Dipartimento di Scienze Economico-Sociali e Matematico-Statistiche, Universit\`a di Torino and Collegio Carlo Alberto, Torino, Italy.}
}

\newif\ifarxiv
\arxivtrue          

\newif\ifshowfigures

\ifarxiv
  \showfigurestrue
  \newcommand{\appsec}{Appendix}
\else
  \showfiguresfalse
  \newcommand{\appsec}{Supplementary Section}
\fi

\usepackage{amsmath, amssymb, amsthm} 
\usepackage{stmaryrd,scalerel}

\newtheorem{theorem}{Theorem}

\renewcommand{\P}[1]{\mathbb{P}\left[#1\right]}

\newcommand{\I}[1]{\mathbb{I}\left[#1\right]}

\usepackage{natbib}
\usepackage{url}

\makeatletter
\providecommand{\doi}[1]{}
\providecommand{\url}[1]{}
\providecommand{\href}[2]{#2}
\providecommand{\eprint}[1]{}
\makeatother

\usepackage{placeins}

\newenvironment{extract}
  {\begin{quote}\itshape}
  {\end{quote}}

\usepackage{comment}
\excludecomment{textbox}
\excludecomment{marginnote}

\newenvironment{summary}[1][]{\section*{#1}}{}
\newenvironment{issues}[1][]{\section*{#1}}{}

\begin{document}

\maketitle

\begin{abstract}
Predictive inference is a fundamental task in statistics, traditionally addressed using parametric assumptions about the data distribution and detailed analyses of how models learn from data. In recent years, conformal prediction has emerged as an alternative framework that is well suited to modern applications involving high-dimensional data and complex machine learning models. Its appeal stems from being both distribution-free---relying mainly on symmetry assumptions such as exchangeability---and model-agnostic, treating the learning algorithm as a black box. Even under such limited assumptions, conformal prediction provides exact finite-sample guarantees, although these are typically marginal and require careful interpretation. This paper explains the core ideas of conformal prediction and reviews selected methods. Rather than offering an exhaustive survey, it aims to provide a clear conceptual entry point and a pedagogical overview of the field.


\end{abstract}


\section{INTRODUCTION}

A pioneering work in conformal prediction \citep{vovk1999machine} opens with:
\begin{extract}
``two important differences of most modern methods of machine learning from classical statistical methods are that: (1) machine learning methods produce bare predictions, without estimating confidence in those predictions; and (2) many machine learning methods are designed to work under the general i.i.d.~assumption and they are able to deal with extremely high-dimensional hypothesis spaces.''
\end{extract}
This observation remains relevant today and explains the growth of conformal prediction, a versatile statistical framework developed to quantify uncertainty in the predictions of complex models while providing exact statistical guarantees under limited assumptions.

The roots of conformal prediction can be traced back at least to foundational statistical work from the 1940s \citep{wilks1941determination,wald1943extension,scheffe1945non,tukey1947non,tukey1948nonparametric}. Remarkably, its core ideas have remained largely unchanged, even as methodology and the scope of machine learning applications have expanded dramatically.
This stability reflects the reliance of conformal prediction on fundamental statistical principles: it requires neither parametric assumptions on the data distribution nor knowledge of the internal mechanics of the predictive models it accompanies. As a result, conformal prediction is well positioned to remain relevant as data sets grow and models continue to evolve.

Many high-quality expository resources on conformal prediction already exist, including books \citep{vovk2022algorithmic,angelopoulos2024theoretical}, literature surveys \citep{shafer2008tutorial,tian2022methods,fontana2023conformal,zhou2025conformal}, and practitioner-oriented tutorials \citep{angelopoulos2023conformal}. Moreover, research in this area is still expanding rapidly, so that a new comprehensive survey of recent advances would risk becoming quickly outdated. Accordingly, the goal of this review is to complement existing resources by offering a concise and pedagogical introduction to some of the key ideas, starting from the beginning and adopting a perspective that should resonate with statisticians.

\begin{textbox}[b]
\section{Significance Statement}
Modern data analysis increasingly relies on complex machine learning models whose predictions may often be accurate but typically lack clear measures of reliability. Conformal prediction addresses this issue by providing principled uncertainty estimation for any predictive model. Rather than depending on assumptions about model structure or data distribution, conformal prediction is based primarily on observed model performance on data representative of the population of interest. This paper outlines the core ideas of conformal prediction, discusses its strengths and limitations, and highlights some useful extensions.
\end{textbox}

\section{FOUNDATIONS} \label{sec:foundations}

\subsection{Exchangeability, Conformal Prediction Sets and $p$-Values} \label{sec:conformal-fundamentals}

\subsubsection{Exchangeable data}
We consider data that can be represented as $n+1$ pairs $Z_i=(X_i,Y_i)$, with $X_i \in \mathcal{X}$ and $Y_i \in \mathcal{Y}$, for all $i \in [n+1] := \{1,\ldots,n+1\}$, where $\mathcal{X},\mathcal{Y}$ are measurable spaces and $\mathcal{Z} = \mathcal{X} \times \mathcal{Y}$. Intuitively, $Y$ is the outcome (or label) to predict and $X$ are features (or covariates) that may be relevant.
Given data $\mathbf{Z}_{1:n} := (Z_1, \ldots, Z_n)$ and new features $X_{n+1}$, the goal is to predict, with confidence, the outcome $Y_{n+1}$, before observing it.
The main assumption is that $Z_{1},\ldots,Z_{n+1}$ are exchangeable random samples from some population.
Throughout the paper, we write $\{Z_{1}, \ldots, Z_{n}\}$ to indicate the unordered set (or bag) of values in $\mathbf{Z}_{1:n}$, preserving possible repetitions, and sometimes shorten this to $\Lbag \mathbf{Z}_{1:n} \Rbag$, consistent with the notation of \citet[Section~2.2.1]{vovk2022algorithmic}.

\begin{marginnote}[]
\entry{Exchangeability}{Random variables $Z_1, \ldots, Z_{n+1}$ are exchangeable if their joint distribution does not change when the indices are permuted. This is a weaker condition than the  classical independent and identically distributed (i.i.d.)\ assumption.}
\entry{Why ``conformal''}{
The term reflects the idea of assessing how well a new sample conforms to previously observed data assumed to be exchangeable.
}
\end{marginnote}

\begin{textbox}[b]
\section{Motivating example: reference ranges for clinical laboratory tests}

Clinical test results are often interpreted using reference ranges that vary with patient characteristics.
For example, let $Y_i$ denote serum creatinine level (a standard marker of kidney function) for patient $i$, and let $X_i$ include features such as age and sex.
Given independent samples $\mathbf{Z}_{1:n} = \{(X_i, Y_i)\}_{i=1}^{n}$ from a healthy reference population, the goal is to construct a patient-specific interval $C_\alpha(X_{n+1}; \mathbf{Z}_{1:n})$ for a new patient.

For a 65-year-old woman, conformal prediction may yield $C_{0.05}(X_{n+1}; \mathbf{Z}_{1:n}) = [0.5,1.2]$ mg/dL.
Marginal coverage guarantees that, on average, no more than $5\%$ of individuals from the reference population have values outside their intervals.
If her measured value is $1.4$ mg/dL, it may warrant further clinical evaluation.
\end{textbox}

\subsubsection{Uncertainty quantification via prediction sets} \label{sec:uq-pred-sets}
Conformal methods quantify uncertainty about $Y_{n+1}$ by constructing a prediction set. This, denoted by $C_\alpha(X_{n+1}; \mathbf{Z}_{1:n}) \subseteq \mathcal{Y}$, may depend on the features $X_{n+1}$, the data $\mathbf{Z}_{1:n}$, and a significance level $\alpha \in (0,1)$.
In practice, $C_\alpha(X_{n+1}; \mathbf{Z}_{1:n})$ is often designed to guarantee marginal coverage:\footnote{This corresponds to one of two coverage requirements for tolerance regions in \citet{wilks1941determination}.}
\begin{align} \label{eq:coverage}
\P{Y_{n+1} \in C_\alpha(X_{n+1}; \mathbf{Z}_{1:n})} \geq 1 - \alpha.
\end{align}
Crucially, this guarantee is finite-sample and distribution-free, meaning it holds exactly for any data distribution under which $\mathbf{Z}_{1:(n+1)}$ are exchangeable.
The probability is taken over $(X_{n+1}, Y_{n+1})$ and $\mathbf{Z}_{1:n}$, all of which are random---hence the term ``marginal''.
Therefore, Eq.~\ref{eq:coverage} does not imply coverage conditional on any particular value of $X_{n+1}$, $Y_{n+1}$, or $\mathbf{Z}_{1:n}$.

\begin{marginnote}[]
\entry{Prediction vs.\ confidence sets}{
Prediction sets may resemble confidence sets, whose aim is $\P{\theta \in C_\alpha(\mathbf{Z}_{1:n})} \ge 1-\alpha$ for a fixed population parameter $\theta$.
Unlike $\theta$, however, $Y_{n+1}$ is random and may be observed in the future.
}
\end{marginnote}

Marginal coverage is a reasonable objective because it is easy to achieve in finite samples under limited assumptions. However, it is rarely fully satisfactory on its own, and that is why practical conformal prediction methods are typically designed to not only satisfy Eq.~\ref{eq:coverage} but also produce prediction sets that are as informative as possible.

Although different applications may call for different measures of informativeness, a broadly appealing ideal goal would be to minimize the average ``size'' (e.g., cardinality or volume) of the prediction sets while achieving feature-conditional coverage,
\begin{align} \label{eq:coverage-X-cond}
\P{Y_{n+1} \in C_\alpha(X_{n+1}; \mathbf{Z}_{1:n}) \mid X_{n+1} = x} \ge 1-\alpha,
\qquad \forall x \in \mathcal{X}.
\end{align}
Prediction sets satisfying Eq.~\ref{eq:coverage-X-cond} with minimal size would be highly informative, with uncertainty tailored to the inherent difficulty of predicting $Y_{n+1}$ given $X_{n+1}$. Unfortunately, guaranteeing feature-conditional coverage with reasonably-sized prediction sets is often impossible, especially when the feature space $\mathcal{X}$ is large \citep[e.g.,][]{vovk2012conditional,lei2014distribution,foygel2021limits}. Consequently, conformal methods are typically designed to seek informative and feature-adaptive prediction sets, while guaranteeing marginal coverage.

\subsubsection{Prediction sets from tests of exchangeability} \label{sec:exchangeability-test}

Conformal prediction operates by building and inverting a test for the null hypothesis $\mathcal{H}_{n+1}$ that the full dataset $Z_1, \ldots, Z_{n+1}$ is exchangeable.
 Since the label $Y_{n+1}$ is unobserved and thus plays a distinguished role, it is helpful to emphasize the dependence of various quantities of interest on its hypothetical value $\tilde{y} \in \mathcal{Y}$. We introduce a function $p:(\tilde{y}; \mathbf{Z}_{1:n}, X_{n+1}) \mapsto [0,1]$ whose role is to quantify the evidence against $\mathcal{H}_{n+1}$ contained in the hypothesized unordered dataset
\[
D(\tilde{y}) := \{Z_1,\ldots,Z_n,(X_{n+1},\tilde{y})\}.
\]
This $p$-function is designed, as detailed below, in such a way that evaluating it at the true (random) test label $Y_{n+1}$ gives a conformal $p$-value $p(Y_{n+1}; \mathbf{Z}_{1:n}, X_{n+1})$: a statistic that is marginally super-uniform under $\mathcal{H}_{n+1}$; that is, if the full data are exchangeable,
\begin{align} \label{eq:pval-su}
  \P{ p(Y_{n+1}; \mathbf{Z}_{1:n}, X_{n+1}) \le \alpha } \le \alpha, \qquad \forall \alpha \in (0,1).
\end{align}

The $\alpha$-level conformal prediction set for $Y_{n+1}$ is the set of labels $\tilde{y} \in \mathcal{Y}$ for which the evidence contained in $D(\tilde{y})$ would be insufficient to reject $\mathcal{H}_{n+1}$ at level $\alpha$; i.e.,
\begin{align} \label{eq:conf-pred-set}
C_\alpha(X_{n+1}; \mathbf{Z}_{1:n}) := \left\{ \tilde{y} \in \mathcal{Y} : p(\tilde{y}; \mathbf{Z}_{1:n}, X_{n+1}) > \alpha \right\}.
\end{align}
In other words, $C_\alpha(X_{n+1}; \mathbf{Z}_{1:n})$ is the acceptance region for this test of $\mathcal{H}_{n+1}$.\footnote{Sometimes, conformal prediction is explained as testing ``random hypotheses'' of the type $\mathcal{H}_{n+1}(\tilde{y}): Y_{n+1}=\tilde{y}$, using $p(\tilde{y}; \mathbf{Z}_{1:n}, X_{n+1})$ as a ``$p$-value'' for $\mathcal{H}_{n+1}(\tilde{y})$. However, that interpretation is not entirely rigorous because $\mathcal{H}_{n+1}(\tilde{y})$ is an event, not a hypothesis, and $p(\tilde{y}; \mathbf{Z}_{1:n}, X_{n+1})$ need not be super-uniform under $\mathcal{H}_{n+1}$ for a fixed $\tilde{y}$.
}

Testing $\mathcal{H}_{n+1}$ is a theoretical device in this context; we compute the acceptance region but we never apply the test, for two reasons. Firstly, the test's decision depends on the unobserved $Y_{n+1}$; secondly, the null hypothesis $\mathcal{H}_{n+1}$ is assumed true from the moment one decides to apply conformal prediction, and the goal is not to disprove it.\footnote{Using the conformal $p$-value $p(Y_{n+1}; \mathbf{Z}_{1:n}, X_{n+1})$ to disprove $\mathcal{H}_{n+1}$ is the goal in a different, related class of outlier detection problems, discussed in Section~\ref{sec:outlier-detection}.}
Nonetheless, this imaginary hypothesis test is useful to prove that $C_\alpha(X_{n+1}; \mathbf{Z}_{1:n})$ has valid $1-\alpha$ marginal coverage for $Y_{n+1}$.
In fact, marginal coverage follows directly from Eq.~\ref{eq:pval-su}--\ref{eq:conf-pred-set}, since
$Y_{n+1} \notin C_\alpha(X_{n+1}; \mathbf{Z}_{1:n})$ if and only if $p(Y_{n+1}; \mathbf{Z}_{1:n}, X_{n+1}) \le \alpha$.
This is an instance of the general duality, via test inversion, between hypothesis tests and prediction (or confidence) regions.

If the outcome space $\mathcal{Y}$ is finite, $C_\alpha(X_{n+1}; \mathbf{Z}_{1:n})$ can be constructed in practice by evaluating the $p$-function explicitly for each hypothetical value $\tilde{y}$.
Analytical simplifications are in many cases possible and sometimes necessary, for example if $\mathcal{Y}$ is uncountable.

\subsubsection{Nonconformity scores and conformal $p$-functions} \label{sec:nonconf-scores}

The $p$-function compares the nonconformity of the hypothesized case $(X_{n+1},\tilde{y})$ with that of the other observations $Z_i \in D(\tilde{y})$, for $i \in [n]$,
where in each case nonconformity is measured relative to $D(\tilde{y})$. This is quantified by a score function
$s: \mathcal{Z} \times \mathcal{Z}^{n+1} \to \mathbb{R}$, invariant to permutations of its second argument
and designed so that $s(z; D(\tilde{y}))$ assigns larger values to more atypical observations $z \in \mathcal{Z}$.

The conformal $p$-function at $\tilde{y}$ is then defined as the relative rank of the hypothesized test score $s((X_{n+1},\tilde{y}); D(\tilde{y}))$ among the scores $s((X_{i},Y_{i}); D(\tilde{y}))$ of all reference observations:
\begin{align} \label{eq:def-pfunc}
p(\tilde{y}; \mathbf{Z}_{1:n}, X_{n+1})
= \frac{1 + \sum_{i=1}^n \I{ s\bigl((X_{n+1},\tilde{y}); D(\tilde{y})\bigr)
\le s\bigl(Z_i; D(\tilde{y})\bigr) }}{n+1}.
\end{align}
This takes values in $\{1/(n+1), 2/(n+1), \ldots, 1\}$, with smaller values for more unusual hypothesized test cases $(X_{n+1},\tilde{y})$.

It is readily verified that this construction yields valid conformal $p$-values.
\begin{theorem} \label{thm:exch-super-unif}
If $\mathbf{Z}_{1:(n+1)}$ are exchangeable,
$ \P{ p(Y_{n+1}; \mathbf{Z}_{1:n}, X_{n+1}) \le \alpha } \le \alpha$, $\forall \alpha \in (0,1)$.
\end{theorem}
\begin{proof}
\sloppy Because $D(Y_{n+1})$ is invariant to permutations, under $\mathcal{H}_{n+1}$ the scores $s(Z_1; D(Y_{n+1})), \ldots, s(Z_n; D(Y_{n+1})), s((X_{n+1},Y_{n+1}); D(Y_{n+1}))$ are exchangeable.
Consequently, if the scores are almost-surely (a.s.)\ distinct, the rank of the last among the $n+1$ values is uniformly distributed on $\{1,\ldots,n+1\}$, which implies Eq.~\ref{eq:pval-su}.
In general, a careful definition of the rank is needed to ensure exact uniformity in the presence of ties. However, conditional on the bag of scores, $s((X_{n+1},Y_{n+1}); D(Y_{n+1}))$ is still uniformly distributed over the bag, which implies that the $p$-value is valid.
\end{proof}

\subsubsection{Connection to pivots} \label{sec:foundations-pivots}

Exchangeability connects conformal prediction to the classical notion of pivots, and specifically also to rank and permutation tests \citep{kuchibhotla2020exchangeability}.
In general, pivots are useful because they can be inverted  to construct confidence sets, and this forms the basis of many common intervals, such as the $t$-interval.
\begin{marginnote}[]
\entry{Pivot}{
Function of data and population parameters whose distribution is known even if the population itself is unknown; see e.g., \cite{cox1979theoretical}. It can be inverted to construct confidence or prediction sets.
}
\end{marginnote}

In conformal prediction, the vector $\mathbf{Z}_{1:(n+1)}$ is a conditional pivot, whose distribution given the bag $\Lbag \mathbf{Z}_{1:(n+1)} \Rbag$ does not depend on any population parameters and is fully known: it is uniform over all permutations of the values in the bag.
The inversion of this conditional pivot, which depends on the unknown outcome $Y_{n+1}$ through $Z_{n+1}$, is precisely what gives the conformal prediction set for $Y_{n+1}$.

This perspective can be pushed further to construct joint coverage regions simultaneously for parameters and outcomes \citep{dobriban2023joint}.

\subsubsection{Marginal coverage: strengths and limitations}
We have seen how the marginal coverage of conformal prediction generally enjoys a finite-sample lower bound. There is also an (almost) matching upper bound, as long as the nonconformity scores are almost-surely distinct. 
The latter is a very mild assumption because any ties can always be broken at random, independent of the data.
This result has a long history and versions of it have appeared in \citet{wilks1941determination,vovk1999machine,vovk2022algorithmic}, and \citet{lei2013distribution}.
\begin{marginnote}[]
\entry{Tie-breaking}{
Nonconformity scores can typically be assumed to be almost surely unique, since ties can be broken by adding independent random noise to the score function.
}
\end{marginnote}
\begin{theorem} \label{thm:general-conformal}
If $Z_1, \ldots, Z_{n+1}$ are exchangeable, for any score function $s: \mathcal{Z} \times \mathcal{Z}^{n+1} \to \mathbb{R}$ that is invariant to permutations of its second argument, and any $\alpha \in (0,1)$, the conformal prediction sets given by Eq.~\ref{eq:conf-pred-set}--\ref{eq:def-pfunc} have marginal coverage of at least $1-\alpha$.
Moreover, if $s(Z_1;D(Y_{n+1})),\ldots,s(Z_{n+1};D(Y_{n+1}))$ are a.s.~distinct, then
\begin{align} \label{eq:coverage-two-sided}
1-\alpha \le \P{Y_{n+1} \in C_\alpha(X_{n+1}; \mathbf{Z}_{1:n})} \le 1-\alpha + \frac{1}{n+1}.
\end{align}
\end{theorem}
The upper bound in Theorem~\ref{thm:general-conformal} shows that coverage converges to $1-\alpha$ at the fast rate $\mathcal{O}(1/n)$. This is very efficient relative to most estimation tasks, for which typically the error converges no faster than $\mathcal{O}(1/\sqrt{n})$ due to the central limit theorem.

However, not all prediction sets satisfying Eq.~\ref{eq:coverage-two-sided} are equally informative. For instance, marginal coverage can be achieved by the trivial prediction set
\begin{align} \label{eq:pred-set-trivial}
  C_\alpha^{\text{trivial}}(X_{n+1}; \mathbf{Z}_{1:n})
  = \begin{cases}
    \mathcal{Y}, & \text{if } \xi_{n+1} \le 1-\alpha, \\
    \emptyset, & \text{otherwise},
  \end{cases}
\end{align}
where the features $X$ are augmented with independent noise $\xi \sim \mathrm{Uniform}(0,1)$, to enable randomization. Despite having exact $1-\alpha$ coverage, this set is uninformative.

This counterexample shows why marginal coverage is not the ultimate goal of conformal prediction. It is rather a basic sanity check, while the real challenge is to construct informative prediction sets. Achieving this often requires careful design of the nonconformity score and, sometimes, more sophisticated methods. We return to this later.

\subsection{Illustration: Predicting a Continuous Scalar Variable} \label{sec:illustrative-numerical}

To build intuition, we begin by studying a simple problem where the goal is to construct a one-sided prediction interval for a continuous outcome without using feature information.

Suppose the distribution of $Y$ is supported on $\mathcal{Y} = \mathbb{R}$ without point masses, and $X = 1$ almost surely, so the features can be ignored.
We focus on constructing a one-sided prediction interval
$C_\alpha(X_{n+1}; \mathbf{Z}_{1:n}) = (-\infty, U_\alpha(\mathbf{Y}_{1:n})]$, with upper bound
$U_\alpha(\mathbf{Y}_{1:n})$.
The marginal coverage objective is $\P{Y_{n+1} \le U_\alpha(\mathbf{Y}_{1:n})} \ge 1 - \alpha$.
Despite its simplicity, this example captures some of the essential mechanics of conformal prediction.

\begin{textbox}[b]
\section{Motivating example: one-sided clinical reference range for cardiac troponin}
In emergency departments, cardiac troponin is measured to assess myocardial injury.
Among healthy individuals, levels are typically very low or undetectable, and concern arises only when values are unusually high.
Conformal prediction can be used to construct an upper reference bound such that the probability a new healthy patient exceeds this threshold is at most $\alpha$.
\end{textbox}

\subsubsection{Construction using conformal $p$-values}
A natural implementation of the framework from Section~\ref{sec:conformal-fundamentals} uses the score function $s((x,y); D(\tilde{y})) = y$ for $y \in \mathcal{Y}$, ignoring $x$ and $D(\tilde{y})$.
With this choice, the $p$-function (Eq.~\ref{eq:def-pfunc}) reduces to $p(\tilde{y}; \mathbf{Z}_{1:n}, X_{n+1}) = [R(\tilde{y}; \mathbf{Y}_{1:n})+1]/(n+1)$, where $R(\tilde{y}; \mathbf{Y}_{1:n}) = \sum_{i=1}^n \I{ \tilde{y} \le Y_i }$ counts the number of observations in $\mathbf{Y}_{1:n}$ greater than or equal to $\tilde{y}$.
Then, the conformal prediction set (Eq.~\ref{eq:conf-pred-set}) becomes
\begin{align} \notag
  C_\alpha(X_{n+1}; \mathbf{Z}_{1:n})
  &= \{ \tilde{y} : R(\tilde{y}; \mathbf{Y}_{1:n}) \ge \lfloor \alpha (n+1) \rfloor \}
   = (-\infty, U_{\alpha}(\mathbf{Y}_{1:n})],\\
  U_{\alpha}(\mathbf{Y}_{1:n})
  &= \text{the } \lceil (1-\alpha)(n+1) \rceil\text{-th smallest element of } \{Y_1,\ldots,Y_n,+\infty\}.
\label{eq:scalar-U}
\end{align}
This is one of many cases where the construction outlined in Eq.~\ref{eq:conf-pred-set}--\ref{eq:def-pfunc} simplifies analytically.

One may wonder why the score function $s((x,y);D(\tilde{y}))=y$ in this example is independent of $D(\tilde{y})$, a choice that substantially streamlines the construction of the conformal prediction set.
This simplification is possible because the numerical outcomes $Y$ already have a natural ordering.
In Section~\ref{sec:illustrative-categorical}, we will see an example with categorical data where a more complicated score function is needed.

\subsubsection{Quantile-based characterization}
Conformal prediction intervals are sometimes presented from a different perspective.
For any $\tau \in [0,1]$, let $Q(\hat{P}(\mathbf{Y}_{1:n}); \tau)$ denote the $\tau$-quantile of the empirical distribution of $\mathbf{Y}_{1:n}$.
Then, the upper bound $U_{\alpha}(\mathbf{Y}_{1:n})$ in Eq.~\ref{eq:scalar-U} can be equivalently written as
\begin{align*}
  U_{\alpha}(\mathbf{Y}_{1:n})
    = Q\left( \hat{P}(\mathbf{Y}_{1:n}); (1-\alpha)\left(1+1/n\right) \right).
\end{align*}
\begin{marginnote}[]
\entry{Quantiles}{
For any distribution $P$ on $\mathbb{R}$ and scalar $\tau \in [0,1]$, the $\tau$-quantile of $P$ is:
$Q(P; \tau) := \inf \left\{ u \in \mathbb{R} : P(u) \geq \tau \right\}$,
where $P(u)$ is the CDF of $P$ at $u$.
For $\tau > 1$, by convention $Q(P; \tau) = +\infty$.
}
\end{marginnote}
This reveals the close connection in this example between the one-sided conformal prediction interval and an ideal interval with knowledge of the true population distribution $P^*$.
All non-randomized one-sided intervals with valid coverage, and
which may depend on $P^*$,  must include the
 one-sided oracle interval $C^*_\alpha = (-\infty, U^*_\alpha]$, where $U^*_\alpha = Q(P^*; 1-\alpha)$.
 The conformal interval is very similar to the empirical plug-in analogue $Q(\hat{P}(\mathbf{Y}_{1:n}); 1-\alpha)$, except that it
 evaluates the empirical quantile at the slightly inflated level $(1-\alpha)(1+1/n)$.
 The following result on the out-of-sample behavior of empirical quantiles provides a direct justification, from this perspective, for the finite-sample coverage of the conformal method.
\begin{theorem} \label{thm:general-conformal2}
\sloppy If $Y_1, \ldots, Y_{n+1} \in \mathbb{R}$ are exchangeable, then for any $\alpha \in (0,1)$, $\mathbb{P}[Y_{n+1} \le Q( \hat{P}(\mathbf{Y}_{1:n}); (1-\alpha)(1+1/n) ) ] \geq 1-\alpha$.
Moreover, if $Y_1, \ldots, Y_{n+1}$ are a.s.~distinct,
\begin{align*}
1-\alpha \le \P{ Y_{n+1} \le Q\left( \hat{P}(\mathbf{Y}_{1:n}); (1-\alpha)\left(1+1/n\right) \right)} 
\le 1-\alpha + \frac{1}{n+1}.
\end{align*}
\end{theorem}

As $n$ increases, $U_{\alpha}(\mathbf{Y}_{1:n})$ converges almost surely to $U^*_{\alpha}$ by the Glivenko--Cantelli theorem, so the conformal prediction intervals are consistent with the oracle. Classical asymptotic theory would tell us that the empirical cumulative distribution function (CDF) converges to the population CDF at rate $\mathcal{O}(1/\sqrt{n})$. However, conformal prediction achieves the nominal $1-\alpha$ coverage from above at the faster rate $\mathcal{O}(1/n)$.
This highlights an important insight: prediction with marginal coverage can be statistically easier than estimation of population quantiles.

\subsubsection{Empirical study}

\appsec~\ref{app:illustrative-numerical} presents simulation studies illustrating the finite-sample, distribution-free validity and efficiency of one-sided conformal prediction intervals. Across a range of data-generating distributions and sample sizes, conformal intervals maintain nominal coverage while rapidly approaching oracle performance. In contrast, heuristic plug-in approaches and classical parametric methods can substantially under- or over-cover depending on sample size and model misspecification.

\subsection{Illustration: Predicting a Categorical Variable} \label{sec:illustrative-categorical}

We now turn to a second example where, despite the continued absence of informative features, it becomes necessary to use a more sophisticated score function $s(z; D(\tilde{y}))$ that depends nontrivially on its second argument, the reference dataset $D(\tilde{y})$.

In this example the outcome $Y$ is categorical, taking values from a finite dictionary $\mathcal{Y}$ of known cardinality $K \geq 1$, while the features (denoted here as $\xi$) are $\text{Uniform}(0,1)$ random variables, independent of $Y$ and thus uninformative.
Without loss of generality, we can set $\mathcal{Y} = [K] = \{1,\ldots,K\}$, with an arbitrary ordering of the labels.
The goal is to construct a small prediction set $C_\alpha(\xi_{n+1}; \mathbf{Z}_{1:n}) \subseteq [K]$ for $Y_{n+1}$ satisfying marginal coverage at level $1-\alpha$.
Because this problem is more subtle than the one from Section~\ref{sec:illustrative-numerical}, we discuss the ideal oracle approach before explaining the conformal solution.

\begin{textbox}[b]
\section{Motivating example: categorical prediction in differential diagnosis}
In emergency departments, patients with different underlying conditions may present the same initial symptoms (e.g., acute chest pain). In this setting, $Y_i$ denotes the diagnosis for patient $i$, taking values in a finite dictionary $\mathcal{Y}$ (e.g., myocardial infarction, pulmonary embolism, pneumonia, or non-cardiac chest pain), while patient features are ignored. For a new patient, the clinical task is not to identify a single diagnosis immediately, but to construct a small set $C_\alpha \subseteq \mathcal{Y}$ of plausible diagnoses likely to contain the true condition.
\end{textbox}

\subsubsection{The oracle approach} \label{sec:illustrative-categorical-oracle}
Let $P^* = (\pi^*_1, \ldots, \pi^*_K)$ denote the (unknown) population distribution of $Y$, where $\pi^*_y = \P{Y = y}$ for $y \in [K]$. For simplicity, assume these probabilities are distinct, so that the oracle prediction set is uniquely defined. Knowing $P^*$, an oracle could construct the most informative prediction set $C^*_\alpha$ with marginal coverage for $Y_{n+1}$ by selecting the smallest subset of labels whose total probability mass is at least $1 - \alpha$; this is obtained by sorting $(\pi^*_1, \ldots, \pi^*_K)$ in decreasing order.
See \appsec~\ref{app:illustrative-categorical} for details on how these oracle prediction sets can be made even smaller through randomization.

\subsubsection{Oracle-inspired conformal prediction sets} \label{sec:illustrative-categorical-estim}
The oracle sorts the labels based on $P^*$, which is unknown in practice.
This is the key difference between this example and the one from Section~\ref{sec:illustrative-numerical}, where the candidate outcomes were
sorted based on their known numerical values.
This suggests that conformal prediction in this example involves an additional complication: one must design a score function $s(z; D(\tilde{y}))$ that computes and suitably leverages an empirical estimate of $P^*$ using the available data in $D(\tilde{y}) := \{Z_1,\ldots,Z_n,(\xi_{n+1},\tilde{y})\}$.
\begin{marginnote}[]
\entry{Seek the oracle}{To obtain informative prediction sets, conformal prediction is often implemented using scores inspired by how an ideal oracle would behave if it knew the true data-generating distribution.}
\end{marginnote}

For a hypothesized dataset $D(\tilde{y})$, with $\tilde{y} \in [K]$, the corresponding maximum-likelihood (or plug-in)
estimate of $\pi^*_y$ under the general multinomial model, for $y \in [K]$, is
\begin{align*}
  \hat{\pi}_y(D(\tilde{y}))
  & := \frac{1}{n+1} \left( \sum_{i=1}^{n} \I{Y_i=y} + \I{y=\tilde{y}} \right)
   = \frac{n}{n+1} \frac{n_y}{n} + \frac{\I{y=\tilde{y}}}{n+1},
\end{align*}
where $n_y$ counts the observations with label $y$ in the observed data $\mathbf{Y}_{1:n}$.

Since $s\big((\xi,y); D(\tilde{y})\big)$ is intended to quantify the nonconformity of label $y$ relative to $D(\tilde{y})$, it should take smaller values for more frequent labels. This motivates using the negative empirical class probability, $-\hat{\pi}_y(D(\tilde{y}))$, as the main component of the score function.
Additionally, using the uninformative features $\xi$ in $z=(\xi,y)$, we include a small tie-breaking term equal to $-(\xi/2)/(n+1)$, leading to
\begin{align*} 
  s\left( (\xi,y); D(\tilde{y}) \right) = - \hat{\pi}_{y}(D(\tilde{y})) - \frac{\xi/2}{n+1}.
\end{align*}
The tie-breaking term affects the ordering of scores only when labels have identical frequencies within $D(\tilde{y})$, and it is useful to ensure that $s(Z_1;D(Y_{n+1})),\ldots,s(Z_{n+1};D(Y_{n+1}))$ are a.s.~distinct, which by Theorem~\ref{thm:general-conformal} is sufficient 
to achieve $\mathcal{O}(1/n)$-tight marginal coverage at level $1-\alpha$ for any $K$ and $\alpha$. This choice of $s$ leads to a conformal $p$-function of the form
\begin{align*}
  p(\tilde{y}; \mathbf{Z}_{1:n}, \xi_{n+1})
  & = \frac{1}{n+1} \left( 1 + \sum_{k=1}^{n_{\tilde{y}}} k |\Gamma_k| - n_{\tilde{y}} + \sum_{i \in I(\tilde{y})} \I{ \xi_{n+1} \geq \xi_i } \right),
\end{align*}
where $\Gamma_k = \{l \in [K]: n_l = k\}$ are the labels observed exactly $k$ times in $\mathbf{Y}_{1:n}$, and $I(\tilde{y}) = \{i \in [n]: Y_i=\tilde{y}\} \cup \{i \in [n]: n_{Y_i} = n_{\tilde{y}} + 1\}$ is the set of indices corresponding to observations with label $\tilde{y}$ or label frequency $n_{\tilde{y}}+1$.
Although this expression may seem intimidating, it is easy to compute and can be understood by focusing on special cases.

\sloppy For an unseen label $\tilde{y} \in \Gamma_0$, $p(\tilde{y}; \mathbf{Z}_{1:n}, \xi_{n+1}) = ( 1 + \sum_{i : Y_i \in \Gamma_1} \I{ \xi_{n+1} \geq \xi_i } ) / (n+1) \sim \text{Unif} \left( [|\Gamma_1|+1] \right) / (n+1)$. Therefore, $\tilde{y} \in C_\alpha(\xi_{n+1}; \mathbf{Z}_{1:n})$ if and only if $p(\tilde{y}; \mathbf{Z}_{1:n}, \xi_{n+1}) > \alpha$, which in this case requires $|\Gamma_1| \geq \lfloor \alpha (n+1) \rfloor$. This reveals a connection between conformal prediction and the classical Good--Turing estimator of the missing mass \citep{good1953population}; see also \citet{xie2025openset} for a similar connection.

If every observation has label $\tilde{y} \in \Gamma_n$, then $p(\tilde{y}; \mathbf{Z}_{1:n}, \xi_{n+1}) = ( 1 + \sum_{i=1}^{n} \I{ \xi_{n+1} \geq \xi_i } ) / (n+1) \sim \text{Unif} \left( [n+1] \right) / (n+1)$, so $\tilde{y} \in C_\alpha(\xi_{n+1}; \mathbf{Z}_{1:n})$  with probability at least $1-\alpha$.

Although we do not prove it formally here, these conformal prediction sets are asymptotically consistent with a randomized version of the oracle from Section~\ref{sec:illustrative-categorical-oracle}. The argument starts by noting that $\hat{\pi}_y(D(\tilde{y})) \overset{\text{p}}{\to} \pi^*_y$ for all $y,\tilde{y} \in [K]$, by the law of large numbers, from which it follows that $p(\tilde{y}; \mathbf{Z}_{1:n}, \xi_{n+1}) \overset{\text{d}}{\to} p^*(\tilde{y}, \xi) = \sum_{k=r(\tilde{y})+1}^{K} \pi^*_{(k)} + \pi^*_{\tilde{y}} \xi$ as $n \to \infty$, where $\xi \sim \text{Uniform}(0,1)$ independent of everything else and $r(\tilde{y})$ is the rank of $\pi_{\tilde{y}}^*$ among the distinct sorted class probabilities $\pi_{(1)}^*>\cdots>\pi_{(K)}^*$.

\subsubsection{Empirical study}
\appsec~\ref{app:illustrative-categorical} presents simulation studies for this example, demonstrating the finite-sample validity and efficiency of conformal prediction sets under varying degrees of class imbalance. Conformal prediction consistently maintains coverage while rapidly approaching oracle performance as the sample size increases. Classical plug-in and Bayesian approaches, by contrast, lack finite-sample frequentist guarantees and can substantially under- or over-cover depending on sample size and prior misspecification.

\subsection{The Full and Split Conformal Workflows}

In the examples above, the features $X$ were completely uninformative, and thus they were either ignored (Section~\ref{sec:illustrative-numerical}) or used solely to randomly break ties between nonconformity scores (Section~\ref{sec:illustrative-categorical}).
In general, however, features are informative, and therefore an effective nonconformity score function must carefully use them.
This is where machine learning models come into play, and there are two classical approaches to leveraging them.

\subsubsection{Full conformal}

Full conformal prediction is the most direct implementation of the framework described in Section~\ref{sec:conformal-fundamentals}, but also the most computationally expensive. Recall that $s((x,y); D(\tilde{y}))$ in Eq.~\ref{eq:def-pfunc} quantifies how unusual an observation $(x,y)$ appears relative to the hypothesized dataset $D(\tilde{y})$. In principle, $s$ may use the unordered data in $D(\tilde{y})$ in any way, including fitting a predictive model that learns the relation between $X$ and $Y$. The score can then be defined, for example, as a generalized residual comparing $y$ to its model-based prediction. Examples for regression and classification are given later.

\begin{marginnote}[]
\entry{Full (or transductive) conformal prediction}{
The model is refitted using the augmented reference data, separately for each test case and hypothesized label.
}
\end{marginnote}

The example of Section~\ref{sec:illustrative-categorical} is a special case, where a multinomial model is fitted by maximum likelihood. There, re-fitting the model for each hypothesized label $\tilde{y}$ is straightforward, but in general full conformal prediction can be prohibitively costly, especially when the predictive model is complex (e.g., a deep neural network) and $\tilde{y}$ may take uncountably many values. Several techniques are available to make full conformal prediction tractable in certain settings by exploiting model structure, see e.g., \cite{burnaev2014efficiency} and \cite{lei2019fast}. Nonetheless, in many applications, a faster approach is needed.

\subsubsection{Split conformal} \label{sec:split-conformal}

Split conformal prediction \citep{papadopoulos2002inductive} is related to the example from Section~\ref{sec:illustrative-numerical}, where the score function takes the form $s((x,y);D(\tilde{y}))=y$, ignoring $D(\tilde{y})$. In general split conformal prediction, $s((x,y);D(\tilde{y}))=s'(x,y)$, where $s': \mathcal{Z} \to \mathbb{R}$, can still be interpreted as computing generalized residuals, similar to full conformal prediction, but is based on a fixed predictive model, independent of $D(\tilde{y})$.

\begin{marginnote}[]
\entry{Split (or inductive) conformal prediction}{
The model is fitted on separate training data prior to conformal prediction.
}
\end{marginnote}

The term ``split conformal'' is used because in practice one may have only a single dataset (of size $N>n$), which must be randomly partitioned into a training subset (of size $N-n$), used to fit the model defining $s'$, and a calibration subset (of size $n$), used for conformal prediction; see Figure~\ref{fig:1} for a visualization of this workflow. As evaluating the nonconformity scores in this setting does not require the re-fitting of a model for each hypothesized label $\tilde{y}$, the function $p(\tilde{y}; \mathbf{Z}_{1:n}, X_{n+1})$ in Eq.~\ref{eq:def-pfunc} is cheaper to compute. Moreover, the prediction set in Eq.~\ref{eq:conf-pred-set} often admits a closed-form representation, eliminating the need to evaluate $p(\tilde{y}; \mathbf{Z}_{1:n}, X_{n+1})$ for every candidate $\tilde{y}$, as in Section~\ref{sec:illustrative-numerical}. These advantages explain the popularity of split conformal prediction.

\begin{figure}[!htb]
\centering
    \includegraphics[width=0.95\linewidth]{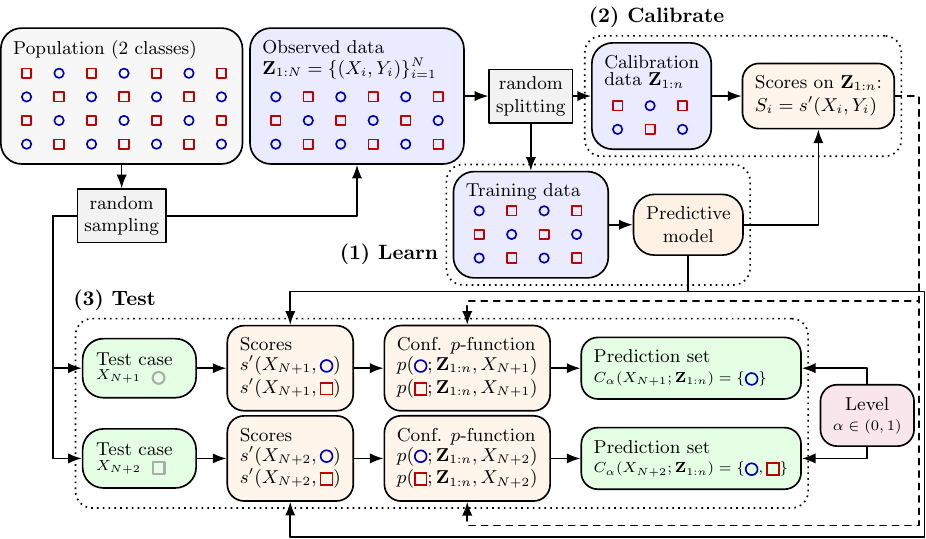}
\caption{
Schematic of split conformal prediction for binary classification.
The data are randomly split into training and calibration subsets.
A predictive model is trained on the training data.
For each test input, nonconformity scores are computed for both hypothesized labels (blue circle and red square).
These scores are compared to the calibration scores to evaluate the conformal $p$-function.
The prediction set comprises labels whose conformal $p$-function exceeds the nominal level $\alpha$.
}
\label{fig:1}
\end{figure}

\subsubsection{Computational and statistical trade-offs}

The computational advantages of split conformal prediction come at the cost of some statistical efficiency. Although marginal coverage is guaranteed for any sample size, the usefulness of the prediction sets depends on the ability of the score function to capture the relation between $X$ and $Y$. Training a model that produces good scores may therefore require substantial data, especially with high-dimensional features. If fewer data are used for training, the learned model may be less accurate, and the resulting prediction sets may be less informative and potentially have lower feature-conditional coverage (Eq.~\ref{eq:coverage-X-cond}) than those from full conformal prediction.

In practice, it is often preferable to allocate most data to model training, since learning the relationship between $X$ and $Y$ is typically the most difficult task. By contrast, marginal coverage converges at rate $\mathcal{O}(1/n)$, suggesting that relatively small calibration samples are typically sufficient; see Section~\ref{sec:calibration-conditional} for more details on how coverage depends on $n$.

\section{METHODOLOGY}

\subsection{Regression}

\subsubsection{Prediction intervals}
To construct two-sided prediction intervals for real-valued outcomes, conformal methods leverage nonconformity scores derived from a regression model trained to approximate the conditional behavior of $Y$ given $X$. The choice of model and score function largely determines the quality of the resulting intervals.

A classical choice is to use a conditional-mean regression model: $s((x,y);D(\tilde{y})) =  |y - \hat{m}(x;D(\tilde{y}))|$, where $\hat{m}(x;D(\tilde{y}))$ is any estimate of $\mathbb{E}[Y \mid X=x]$ \citep{vovk2022algorithmic}. Under the split conformal framework, where $\hat{m}$ is obtained from a distinct training dataset and thus $\hat{m}(x;D(\tilde{y})) = \hat{m}(x)$ does not depend on $D(\tilde{y})$, Eq.~\ref{eq:conf-pred-set} can be equivalently written as
\begin{align} \label{eq:cp-int-regr}
C_\alpha(X_{n+1}; \mathbf{Z}_{1:n})
= \hat{m}(X_{n+1}) \pm Q\left(\hat{P}(\mathbf{S}_{1:n}); (1-\alpha)\frac{n+1}{n}\right),
\end{align}
where $S_i = s'(X_i, Y_i) = |Y_i - \hat{m}(X_i)|$ is the $i$-th residual, for $i \in [n]$.
Although these intervals have nice properties in homoscedastic settings \citep{lei2018distribution}, they have constant width and thus generally lack feature-conditional coverage and adaptivity to heteroscedasticity. This limitation has motivated several alternative score functions.
\begin{marginnote}[]
\entry{Heteroscedasticity}{Feature-specific variability in the spread of the distribution of $Y\mid X$. It makes it more difficult to construct informative prediction sets.}
\entry{Homoscedasticity}{The lack of heteroscedasticity.}
\end{marginnote}

A widely used approach  replaces mean-regression with quantile-based nonconformity scores \citep{romano2019conformalized}. A pair of quantile regression models $\hat{q}_\ell(x; D(\tilde{y}))$ and $\hat{q}_u(x; D(\tilde{y}))$ is trained to approximate the lower and upper $(\alpha/2,\,1-\alpha/2)$ quantiles of $Y$ conditional on $X=x$. The corresponding score function takes the form $s((x,y);D(\tilde{y})) = \max\{\,\hat{q}_\ell(x; D(\tilde{y})) - y,\; y - \hat{q}_u(x; D(\tilde{y}))\,\}$; then, under the split conformal framework, the prediction interval simplifies to
\begin{align} \label{eq:cp-reg-cqr}
& C_\alpha(X_{n+1}; \mathbf{Z}_{1:n}) = \bigl[\hat{q}_\ell(X_{n+1}) - \hat{\tau},\; \hat{q}_u(X_{n+1}) + \hat{\tau}\bigr],
& \hat{\tau} = Q\left(\hat{P}(\mathbf{S}_{1:n}); (1-\alpha)\frac{n+1}{n}\right),
\end{align}
where $S_i = s'(X_i, Y_i)$ for $i \in [n]$. In this case, local adaptivity is provided by the quantile regression models, and marginal coverage by the conformal adjustment $\hat{\tau}$.
\citet{sesia2020comparison} show that, if the model-based conditional quantile estimates are consistent, these prediction intervals asymptotically approach the shortest among equal-tailed prediction intervals with feature-conditional coverage (Eq.~\ref{eq:coverage-X-cond}).

Alternative score functions can yield even more adaptive prediction sets by modeling conditional distributions beyond the mean or specific quantiles \citep{izbicki2019flexible,chernozhukov2021distributional,sesia2021conformal}.

\subsubsection{Empirical example: serum creatinine}

Figure~\ref{fig:2} presents an empirical comparison of conformal prediction intervals for serum creatinine using data from the National Health and Nutrition Examination Survey (NHANES) \citep{paulose2021national}. We focus on an apparently healthy reference population, excluding participants with self-reported kidney disease or pregnancy, and use age and sex as covariates; after removing missing values, the sample size is 6{,}090. The data are randomly split into training (4{,}263), calibration (913), and test (914) sets to implement split conformal prediction and assess performance.
Each observation in the test set separately plays the role of the $(n+1)$-th point, for which a prediction interval for $Y_{n+1}$ is constructed using the corresponding $X_{n+1}$, the model trained on the training data, and the $n=913$ calibration points. To estimate marginal coverage, the coverage indicator for $Y_{n+1}$ is averaged over the 914 test points. The average interval width is estimated similarly.

We compare the two regression-based conformal approaches described above, at level $\alpha = 0.05$. For the mean-based method with absolute residual nonconformity scores (Eq.~\ref{eq:cp-int-regr}), we fit a generalized additive model using the \texttt{gam} package \citep{R-gam} in \textsf{R} \citep{R}, with a smooth age effect and a sex main effect. For the quantile-based method (Eq.~\ref{eq:cp-reg-cqr}), we fit analogous quantile generalized additive models using \texttt{qgam} \citep{R-qgam} to estimate the lower and upper conditional quantiles. 
Both methods
achieve empirical test coverage close to 95\%, with similar average interval lengths, but only the quantile-based approach adapts to age-dependent heteroscedasticity.

\begin{figure}[!htb]
    \centering
    \includegraphics[width=0.8\linewidth]{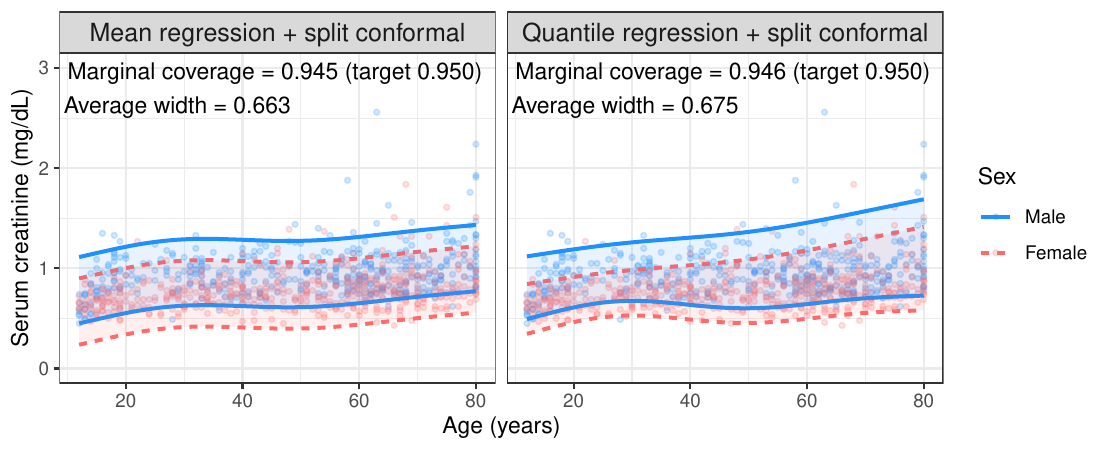}
\caption{
Conformal prediction intervals (\(\alpha = 0.05\)) for serum creatinine as a function of age and sex, using NHANES data restricted to a healthy reference population without self-reported kidney disease or pregnancy.
Dots denote observed outcomes for a hold-out test set, and curves indicate the lower and upper prediction bounds.
Left: intervals based on nonlinear mean regression; right: intervals based on quantile regression.
The quantile-based approach adapts to heteroscedasticity.
}
    \label{fig:2}
\end{figure}

\subsubsection{Beyond prediction intervals}
Several works develop conformal prediction sets that account for multimodal responses \citep[e.g.,][]{lei2013distribution,lei2014distribution,izbicki2022cdsplit}. Others construct prediction sets for multivariate responses \citep{messoudi2021copula,colombo2024normalizing,braun2025minimum,klein2025multivariate,fan2025interpretable}; see \citet{dheur2025unified} for a survey. Methods for the related problem of constructing prediction sets for functions of multiple test cases are developed in \cite{lee2024batch}.

\subsection{Classification}

\subsubsection{Probabilistic models and scores} \label{sec:classification-scores}

To construct prediction sets for $K$-class classification, conformal methods typically leverage a model that estimates conditional class probabilities. For any label $y \in [K]$, let $\hat{\pi}_{y}(x; D(\tilde{y}))$ denote the model's estimate of $\mathbb{P}(Y = y \mid X=x)$, which may be fitted using either the augmented reference dataset $D(\tilde{y})$ (full conformal) or a separate training set (split conformal).
In the former case, $\hat{\pi}_{y}(x; D(\tilde{y}))$ is the natural feature-dependent extension of $\hat{\pi}_{y}(D(\tilde{y}))$ from Section~\ref{sec:illustrative-categorical-estim}.
Most classifiers, including multinomial logistic regression, neural networks with a softmax output layer, boosted trees, and random forests, provide such estimates.

A classical choice of nonconformity score function is
$s((x,y);D(\tilde{y})) = 1-\hat{\pi}_{y}(x;D(\tilde{y}))$ \citep[e.g.,][]{vovk2022algorithmic}.
In the split conformal setting, where
$s((x,y);D(\tilde{y})) = 1 - \hat{\pi}_y(x)$
for a fixed probabilistic classifier, these scores lead to prediction sets of the intuitive form
\begin{align*}
& C_\alpha(X_{n+1}; \mathbf{Z}_{1:n}) = \left\{ \tilde{y} \in [K] : \hat{\pi}_{\tilde{y}}(X_{n+1}) \ge 1-\hat{\tau} \right\},
& \hat{\tau} = Q\!\left(\hat{P}(\mathbf{S}_{1:n});\, (1-\alpha)(1+1/n)\right),
\end{align*}
where $S_i = s'(X_i, Y_i)$ for $i \in [n]$. These prediction sets approximately minimize the expected number of labels included subject to marginal coverage \citep{sadinle2016least}.
One limitation is that they are guaranteed to be non-empty only if $\hat{\tau} \geq (K-1)/K$; otherwise, they are empty for test cases where the estimated label probabilities are all below $1-\hat\tau$. In practice, uninformative empty sets can be avoided by always including at least the label with the smallest nonconformity score.
Another limitation is that they cannot adapt if the conditional distribution of $Y$ given $X$ varies substantially in its concentration across the   feature space, possibly leading to poor feature-conditional coverage \citep{cauchois2021knowing}.

An alternative approach that aims to minimize prediction set size while seeking approximate feature-conditional coverage uses adaptive scores based on cumulative class probabilities \citep{romano2020classification}. This approach sorts the labels in decreasing order of $\hat{\pi}_y(x)$ and includes them in the prediction set until their cumulative probability exceeds a calibrated threshold, leading to smaller sets when the true conditional label distribution is more concentrated. 
Some extensions focus on preventing very large sets when the classifier provides inaccurate probability estimates \citep{angelopoulos2021uncertainty}, and  optimizing the probability of maximally informative singleton prediction sets \citep{wang2026singletonoptimized}.

\subsubsection{Empirical example: diabetes classification} \label{sec:example-diabetes}

Figure~\ref{fig:3} illustrates split conformal prediction for binary
classification using NHANES data, where the outcome is diabetes status
determined by self-report and standard laboratory criteria. After
preprocessing (\appsec~\ref{app:diabetes_details}), 2{,}125 patients over
30 years of age remain, of whom 431 (20.3\%) are labeled as diabetic. We
split these at random into training (1{,}062), calibration (319), and test
(744) sets.

On the training set we fit a logistic regression model to estimate the conditional
probability of diabetes given demographic variables (age, sex,
race/ethnicity, poverty-income ratio), anthropometric measures (waist
circumference, height), cardiometabolic markers (systolic blood pressure,
triglyceride-to-HDL ratio, ALT, uric acid, GGT), and behavioral variables
(sleep duration, self-reported physical activity). We then construct a conformal
prediction set for each test patient at level $\alpha = 0.05$,
using the classical nonconformity score function described in
Section~\ref{sec:classification-scores} and obtaining the calibrated threshold $\hat \tau \approx 0.813$.

Among the 744 test patients, 417 receive the singleton prediction set $\{\text{Healthy}\}$, with an error rate of 6.7\%. Four receive $\{\text{Diabetes}\}$, with no errors. The remaining 323 receive the two-label set $\{\text{Healthy},\text{Diabetes}\}$, reflecting uncertainty; coverage within this group is trivially 100\%.
Overall, the empirical coverage is approximately 96\%.

\begin{figure}[!htb]
\centering
\includegraphics[width=\linewidth]{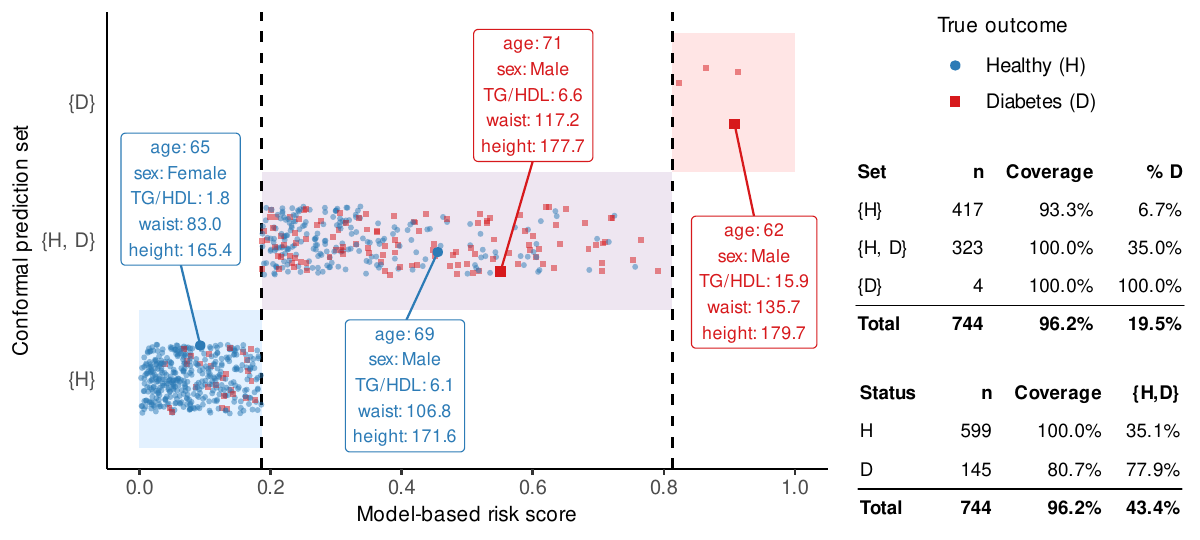}
\caption{
Split conformal prediction sets ($\alpha = 0.05$) for diabetes classification using NHANES data (patients aged over 30).
Test patients are plotted by their model-based predicted probability of diabetes (x-axis).
Dashed vertical lines indicate the three distinct prediction regions: \{Healthy\}, \{Healthy, Diabetes\}, and \{Diabetes\}.
Dots are colored and shaped by the true outcome, and shaded bands denote the proportion of true diabetes cases in each region.
The test coverage is empirically at the desired level. A few important features for four representative patients are highlighted.
}
\label{fig:3}
\end{figure}

\subsubsection{Beyond standard classification}

Conformal prediction extends beyond standard classification to more complex settings.
These include open-set classification, where test cases may belong to unseen classes and prediction sets must capture novelty \citep{xie2025openset}, as well as structured tasks including multi-label classification \citep{papadopoulos2014cross,wang2015comparison,lambrou2016binary,cauchois2021knowing}, where instances may have multiple labels, and hierarchical classification, where labels form a taxonomy and prediction sets must remain semantically coherent \citep{mortier2025conformal}.

\subsection{Outlier Detection} \label{sec:outlier-detection}

In the applications described above, testing the exchangeability of $\mathbf{Z}_{1:(n+1)}$ is a device for constructing a prediction set for the future outcome $Y_{n+1}$. By contrast, in outlier detection (or anomaly detection) applications, $Z_{n+1}$ is observed, and testing the null hypothesis $\mathcal{H}_{n+1}$ of exchangeability is the primary objective.

Consider for example fraud detection, where $\mathbf{Z}_{1:n}$ represent historical legitimate transactions sampled from a stable distribution and $Z_{n+1}$ is a new transaction that must be validated (if exchangeable with $\mathbf{Z}_{1:n}$) or flagged for further review (if deemed non-exchangeable). A natural statistical goal is to maximize the detection of fraudulent transactions (i.e., minimize type-II error) while controlling the rate of false positives (type-I error). Conformal $p$-values can directly address this problem.

\subsubsection{Testing exchangeability with conformal $p$-values}

Since $\mathbf{Z}_{1:(n+1)}$ are fully observed in this setting, the conformal $p$-function (Eq.~\ref{eq:def-pfunc}) can be evaluated at the true  value $Y_{n+1}$ instead of using a fixed hypothesized label, yielding the conformal $p$-value
\begin{align} \label{eq:def-pvalue-outlier}
p(Z_{n+1}; \mathbf{Z}_{1:n})
= \frac{1 + \sum_{i=1}^n \I{ s(Z_{n+1}; \Lbag \mathbf{Z}_{1:(n+1)} \Rbag) \le s(Z_i; \Lbag \mathbf{Z}_{1:(n+1)} \Rbag)} }{n+1}.
\end{align}
Here, $s(z; \Lbag \mathbf{Z}_{1:(n+1)} \Rbag)$ is a nonconformity score designed to quantify how atypical an observation $z$ is relative to the unordered reference set $\Lbag \mathbf{Z}_{1:(n+1)} \Rbag$. Therefore, the second argument of $s$ is invariant to permutations, and the conformal $p$-value is super-uniform (cf.\ Eq.~\ref{eq:pval-su}) under the null hypothesis $\mathcal{H}_{n+1}$ that $\mathbf{Z}_{1:(n+1)}$ are exchangeable \citep{vovk2022algorithmic}. Consequently, rejecting $\mathcal{H}_{n+1}$ whenever $p(Z_{n+1}; \mathbf{Z}_{1:n}) \le \alpha$ yields a valid level-$\alpha$ test.

Many nonconformity score functions for outlier detection are possible, including distances to nearest neighbors, likelihood or density estimates, reconstruction errors from autoencoders, and scores derived from one-class classifiers such as one-class SVMs or isolation forests. In each case, the underlying model may be trained either on an independent dataset (split conformal) or on the augmented reference dataset $\Lbag \mathbf{Z}_{1:(n+1)} \Rbag$ (full conformal).

\subsubsection{Multiple testing with conformal $p$-values}

In many applications, one observes $m$ test cases $Z_{n+1},\ldots,Z_{n+m}$ that must be simultaneously screened for outliers, leading to a multiple testing problem where each hypothesis $\mathcal{H}_{n+j}$ asserts that $Z_{n+j}$ is exchangeable with the reference sample $\mathbf{Z}_{1:n}$. One may want to test a global null (e.g., whether all new observations are exchangeable), identify likely outliers while controlling the false discovery rate (FDR), or perform more structured inference. These questions have attracted recent interest, partly due to the nontrivial dependence among conformal $p$-values $p(Z_{n+1}; \mathbf{Z}_{1:n}), \ldots, p(Z_{n+m}; \mathbf{Z}_{1:n})$ that share the same reference sample.

\citet{bates2023testing} show that conformal $p$-values can be combined with classical multiple testing procedures, including the Benjamini--Hochberg procedure \citep{benjamini1995controlling}, because they satisfy positive regression dependency on a subset (PRDS) \citep{benjamini2001control}. 
Subsequent work characterizes the joint distribution of conformal $p$-values \citep{gazin2024transductive} and extends conformal methods for multiple testing. These include learning more powerful score functions via positive-unlabeled learning \citep{marandon2024adaptive}, leveraging labeled outliers through adaptive weighting of conformal $p$-values to improve power \citep{liang2022integrative}, and developing procedures beyond FDR control for global testing and outlier enumeration \citep{magnani2023collective}.

\subsection{Other Supervised Learning Tasks}

Conformal prediction applies to many additional supervised learning problems. In matrix completion, it can produce uncertainty sets for missing entries, either individually \citep{gui2023conformalized} or jointly across related entries \citep{liang2024structured}. In trajectory forecasting, it constructs prediction bands with simultaneous coverage over time \citep{stankeviciute2021conformal,lindemann2023safe,lekeufack2024conformal,zhou2024conformalized}. More recently, it has been applied to image segmentation \citep{brunekreef2024kandinsky,mossina2025conformal} and natural language generation \citep{kumar2023conformal,ICLR2024_31421b11,pmlr-v235-mohri24a,cherian2024large,chan2025conformal}.

\section{EXTENSIONS}

\subsection{Beyond Exchangeable Data} \label{sec:weighted-cp}

Many conformal prediction methods rely on the full exchangeability of $\mathbf{Z}_{1:(n+1)}$; however, as anticipated in Section \ref{sec:foundations-pivots}, 
the idea is applicable more generally, whenever conditional pivots are available. 
\citet{tibshirani2019conformal} leverage this flexibility by assuming that, conditional on the bag $\Lbag \mathbf{Z}_{1:(n+1)} \Rbag$, we know 
how likely each observation is to occupy the role of the test case. 
With this conditional pivot, conformal predictions can be obtained by appropriately reweighting the observations in the conformal $p$-function (Eq.~\ref{eq:def-pfunc}).

Formally, let $f$ denote the joint law of $(Z_1,\ldots,Z_{n+1})$, possibly known only up to a proportionality constant, and interpreted as a probability mass function for discrete data or a density for continuous data. Let $\Pi_{n+1}$ denote the set of all permutations of $[n+1]$. For each $i \in [n+1]$ and  $z=(z_1,\ldots,z_{n+1}) \in \mathcal{Z}^{n+1}$, define the weight
\begin{align} \label{eq:def-weights}
  w_i^f(z)
  := \frac{\sum_{\sigma \in \Pi_{n+1}:\,\sigma(n+1)=i}
  f(z_{\sigma(1)},\ldots,z_{\sigma(n+1)})}
  {\sum_{\sigma \in \Pi_{n+1}}
  f(z_{\sigma(1)},\ldots,z_{\sigma(n+1)})}.
\end{align}
Intuitively, $w_i^f(z)$ is the probability that $Z_{n+1} = z_i$ conditional on $\Lbag \mathbf{Z}_{1:(n+1)} \Rbag = \Lbag z_{1:(n+1)} \Rbag$. 
As long as these weights are known---so that $(Z_1,\ldots,Z_{n+1})$ is a conditional pivot---the rank-based logic of conformal prediction extends to non-exchangeable settings.

To shorten the notation, for each candidate label $\tilde{y} \in \mathcal{Y}$ let
$\widetilde{\mathbf{Z}}^{\tilde{y}} = (\widetilde Z_1^{\tilde{y}},\ldots,\widetilde Z_{n+1}^{\tilde{y}})$ denote the augmented sample
defined by $\widetilde Z_i^{\tilde{y}} = Z_i$ for $i \in [n]$ and $\widetilde Z_{n+1}^{\tilde{y}} = (X_{n+1},\tilde{y})$, so that
$D(\tilde{y}) = \{ \widetilde Z_1^{\tilde{y}},\ldots,\widetilde Z_{n+1}^{\tilde{y}} \}$.
Then the pivotal approach from Section \ref{sec:foundations-pivots} reduces to defining a weighted $p$-function as
\begin{align} \label{eq:def-pvalue-w}
p^f(\tilde{y}; \mathbf{Z}_{1:n}, X_{n+1})
:= \sum_{i=1}^{n+1} w_i^f(\widetilde Z^{\tilde{y}})\,
\I{ s\bigl(\widetilde Z_{n+1}^{\tilde{y}};D(\tilde{y})\bigr)
   \le s\bigl(\widetilde Z_i^{\tilde{y}};D(\tilde{y})\bigr) }.
\end{align}
If $f$ correctly specifies the joint law of $(Z_1,\ldots,Z_{n+1})$ up to normalization, evaluating this function at the true label yields a valid conformal $p$-value.

\begin{theorem}\label{thm:weighted-super-unif}
Assume that $(Z_1,\ldots,Z_{n+1})$ has joint law $f$ on $\mathcal{Z}^{n+1}$ (known up to a constant). Define the weighted conformal $p$-function $p^f(\tilde{y};\mathbf{Z}_{1:n},X_{n+1})$ as in Eq.~\ref{eq:def-pvalue-w}.
Then,
\begin{align} \label{eq:pval-w-su}
  \P{ p^f(Y_{n+1}; \mathbf{Z}_{1:n}, X_{n+1}) \le \alpha } \le \alpha, \qquad \forall \alpha \in (0,1).
\end{align}
\end{theorem}

\begin{proof}
Let $\mathbf{Z}=(Z_1,\ldots,Z_{n+1})$ and $D := D(Y_{n+1}) =\Lbag \mathbf{Z}_{1:(n+1)} \Rbag$. Under $f$, the conditional probability that $Z_i$ occupies the last position given $D$ is precisely $w_i^f(\mathbf{Z})$. Hence,
\[
p^f(Y_{n+1};\mathbf{Z}_{1:n},X_{n+1}) \overset{d}{=} \P{ V \leq V' \mid D, V },
\]
with $V = s(Z_I;D)$ and $V' = s(Z_J;D)$, where $I$ is a random index drawn according to $\mathbb{P}(I=i \mid D)=w_i^f(\mathbf{Z})$ and $J$ is an independent draw from the same conditional distribution.
Since $V,V'$ are i.i.d.~conditional on $D$, $\mathbb{P}(V\le V'\mid D, V)$ is super-uniform, from which the result follows.
\end{proof}

From Theorem~\ref{thm:weighted-super-unif}, conformal prediction sets are constructed as in the exchangeable setting. Specifically, the $\alpha$-level conformal prediction set for $Y_{n+1}$ comprises all candidate labels $\tilde{y} \in \mathcal{Y}$ for which $p^f(\tilde{y}; \mathbf{Z}_{1:n}, X_{n+1})$ is larger than $\alpha$, analogously to Eq.~\ref{eq:conf-pred-set},
\begin{align} \label{eq:conf-pred-set-weighted}
C^f_\alpha(X_{n+1}; \mathbf{Z}_{1:n})
:= \left\{ \tilde{y} \in \mathcal{Y} : p^f(\tilde{y}; \mathbf{Z}_{1:n}, X_{n+1}) > \alpha \right\}.
\end{align}
This can be interpreted as the acceptance region for a test of the null hypothesis that the joint distribution of $\mathbf{Z}_{1:(n+1)}$ is correctly specified by the function $f$, up to a constant.

\paragraph*{Special case: exchangeable data.}
For any $f$ that is invariant under permutations---that is, under which $Z_1,\ldots,Z_{n+1}$ are exchangeable---$w_i^f(z)=1/(n+1)$ for all $i$, and substituting these weights into Eq.~\ref{eq:def-pvalue-w} recovers the conformal $p$-function from Eq.~\ref{eq:def-pfunc}. Classical conformal prediction therefore re-appears as a special case of the weighted framework.

\paragraph*{Special case: covariate shift.}
An important departure from exchangeability is covariate shift, where the distribution of $Y$ conditional on $X$ is the same for all $n+1$ observations, but the marginal feature distribution differs between $\mathbf{X}_{1:n}$ and $X_{n+1}$. 
This setting arises when reference and test populations differ but the predictive relationship remains stable. For example, a model trained to predict diabetes from demographic and clinical features in one population may be deployed in another with a different age or lifestyle distribution. Although the feature distribution shifts, the conditional relationship with diabetes, reflecting underlying biological mechanisms, may remain unchanged.
Let $P_X$ denote the marginal distribution of $X$ for the first $n$ observations and $Q_X$ that of $X_{n+1}$.
Conformal prediction can be applied under covariate shift as long as the shift is not so extreme as to require extrapolation, as we must assume that $Q_X$ places no mass outside the support of $P_X$.

Under covariate shift, the joint data distribution factorizes as $f(Z_1,\ldots,Z_{n+1}) \propto \prod_{i=1}^{n+1} p(Y_i \mid X_i)\, p_i(X_i)$, where $p(Y\mid X)$ is common to all observations, $p_i(X)=P_X(X)$ for $i\le n$, and $p_{n+1}(X)=Q_X(X)$. Substituting this factorization into the definition of the weights in Eq.~\ref{eq:def-weights} yields
\[
w_i^f(\mathbf{Z}_{1:n},(X_{n+1},\tilde{y}))
=
\dfrac{\mathrm{d}Q_X}{\mathrm{d}P_X}(X_i)
/\left(\sum_{j=1}^{n+1} \dfrac{\mathrm{d}Q_X}{\mathrm{d}P_X}(X_j)\right),
\qquad i\in[n+1],
\]
where $\mathrm{d}Q_X/\mathrm{d}P_X$ is the density ratio between $Q_X$ and $P_X$.

Consequently, the conformal $p$-function in Eq.~\ref{eq:def-pvalue-w} calculates a reweighted rank statistic, where observations contribute according to how representative their features are of the test distribution. In practice, $\mathrm{d}Q_X/\mathrm{d}P_X$ may be unknown but can be estimated when enough samples from $Q_X$ are available, for example by fitting a binary classifier to distinguish training from test covariates \citep{tibshirani2019conformal}.
Coverage guarantees with estimated weights are given in \cite{yang2024doubly}. A method that avoids direct weight estimation, and is better suited to high dimensional settings, is proposed in \cite{joshi2025conformal}.

\paragraph*{Other special cases.}

Although Eq.~\ref{eq:def-weights} involves an apparently daunting sum over an exponential number of permutations, there are many other non-exchangeable settings where weighted conformal prediction can be applied. A prominent example is label shift \citep{podkopaev2021distribution}, where the class-conditional distributions of $X \mid Y$ remain unchanged but the marginal label distribution differs across $\mathbf{Z}_{1:n}$ and $Z_{n+1}$.

Weighted conformal prediction also applies to more structured sampling schemes. For instance, \citet{liang2024structured} consider sampling without replacement from a finite population, where $\mathbf{Z}_{1:n}$ are exchangeable among themselves but not with $Z_{n+1}$. \citet{xie2025openset} consider stratified sample splitting to address class imbalance in classification.

\subsection{Beyond Marginal Coverage}

Marginal coverage (Eq.~\ref{eq:coverage}) averages over variation in both the observed data $\mathbf{Z}_{1:n}$ and the test case $Z_{n+1}=(X_{n+1}, Y_{n+1})$. Although this is appealing for its simplicity, stronger guarantees can be obtained by separating the two sources of randomness, see e.g., \cite{wilks1941determination,vovk2012conditional}. We first consider conditioning on aspects of the test case, relaxing Eq.~\ref{eq:coverage-X-cond} from Section~\ref{sec:foundations}. Conditioning on the calibration data is discussed in Section~\ref{sec:calibration-conditional}.

\subsubsection{Conditioning on the test case}

If the features $X$ are informative about the outcome $Y$, it is natural to ask whether uncertainty guarantees should hold not only on average over test cases, but also conditionally on relevant subsets of the sample space. A general way to formalize this idea is through a class of conditioning events.

Let $\mathcal{G}$ be a collection of measurable functions $g:\mathcal{X}\times\mathcal{Y}\to\{0,1\}$, and consider coverage guarantees of the form
\begin{align} \label{eq:functional-conditional-coverage}
\P{Y_{n+1}\in C_\alpha(X_{n+1};\mathbf{Z}_{1:n}) \mid g(X_{n+1},Y_{n+1})=1 }\ge 1-\alpha,
\qquad \forall g\in\mathcal{G},
\end{align}
whenever the conditioning event has positive probability. This formulation, adopted for example by \citet{gibbs2025conformal}, unifies many useful notions of conditional coverage.

If $\mathcal{G}$ consists of indicator functions of the form $g_{y'}(x,y)=\mathbb{I}\{y=y'\}$, one obtains label-conditional coverage, particularly relevant in classification \citep{vovk2012conditional}.
If instead $\mathcal{G}$ contains indicators of categorical feature values (e.g., a patient's sex when predicting serum creatinine), this yields group-conditional coverage, sometimes motivated by algorithmic fairness considerations \citep{romano2020malice}.
As an extreme case, taking $\mathcal{G}$ to include all singleton feature indicators leads to feature-conditional coverage, as in Eq.~\ref{eq:coverage-X-cond}.
Intermediate choices of $\mathcal{G}$ correspond to weaker conditional coverage notions based on multiple, possibly overlapping neighborhoods or strata \citep{gibbs2025conformal}.

\paragraph{Impossibility results}
As noted in Section~\ref{sec:uq-pred-sets}, exact feature-conditional coverage (Eq.~\ref{eq:coverage-X-cond}) is unattainable without stronger distributional or regularity assumptions. \citet{foygel2021limits} show that any method achieving Eq.~\ref{eq:coverage-X-cond} in finite samples for arbitrary distributions must produce prediction sets so large as to be uninformative.

This difficulty extends to guarantees of the form in Eq.~\ref{eq:functional-conditional-coverage} when the class $\mathcal{G}$ of conditioning events is sufficiently rich. Conformal prediction compares a test case to relevant past observations, and there may be too few of those if conditioning events have low probability. For instance, conditioning simultaneously on sex, age, height, weight, clinical history, and lifestyle may make a patient effectively unique in the NHANES dataset.

In such cases, obtaining informative conformal predictions requires aggregating evidence across similar but non-identical cases, by leaning on (parametric) model assumptions or relaxing the target guarantees. Conformal prediction emphasizes the latter strategy, though model-based approaches remain essential for learning informative nonconformity scores.

\paragraph{Conformal approaches}
Several methods aim to (approximately) achieve conditional guarantees by modifying the definition of conformal $p$-functions. These approaches intervene in the ranking step of Eq.~\ref{eq:def-pfunc}, for example by restricting or reweighting observations.

A simple way to achieve Eq.~\ref{eq:functional-conditional-coverage} applies when $\mathcal{G}$ partitions the sample space into disjoint strata: the $p$-function $p(\tilde{y}; \mathbf{Z}_{1:n}, X_{n+1})$ in Eq.~\ref{eq:def-pfunc} is then computed by restricting the ranking to cases in the same stratum as the test case under the hypothesized label $\tilde{y}$. This is known as Mondrian conformal prediction \citep{vovk2022algorithmic,vovk2012conditional}.

\citet{gibbs2025conformal} extend this idea to overlapping conditioning events, constructing prediction sets that guarantee Eq.~\ref{eq:functional-conditional-coverage} across partially overlapping subpopulations such as ``male'', ``female'', ``under 50'', and ``over 50''.

Alternatively, localized conformal prediction \citep{guan2023localized} modifies the ranking in Eq.~\ref{eq:def-pfunc} by prioritizing observations more similar to the test case. Unlike weighted conformal prediction for non-exchangeable data (Section~\ref{sec:weighted-cp}), the goal here is not to restore marginal validity under distribution shift—exchangeability is still assumed—but to improve conditional adaptivity while retaining finite-sample marginal coverage.

\paragraph{Learning approaches}
A complementary strategy is to simply focus on designing and fitting sufficiently expressive nonconformity score functions, which can lead to highly adaptive and practically informative conformal prediction sets even without finite-sample conditional guarantees.
In regression, \citet{sesia2020comparison} show that quantile-based scores produce conformal prediction intervals (e.g., Eq.~\ref{eq:cp-reg-cqr}) that asymptotically achieve optimal feature-conditional performance if the conditional quantiles are estimated consistently. In classification, \citet{romano2020classification} propose adaptive scores based on cumulative class probabilities, which enjoy similar oracle-consistency conditional properties.

These results suggest that limited data may often be more effectively invested in improving the model underlying the nonconformity score function rather than complicating the calibration step. Conformal prediction aims to separate learning from predictive inference: the first stage fits the score function, while the second converts the scores into prediction sets with formal coverage guarantees. Achieving feature-conditional coverage essentially requires an understanding of the relationship between $X$ and $Y$, and that burden belongs more naturally to the training stage. Sometimes one needs to guarantee coverage conditional on specific features, which, if categorical, may not be too difficult.
However, stronger conditional guarantees require additional calibration data, at the expense of training.

Learning algorithms have been proposed to train score functions using conformal objectives \citep{colombo2020training,bellotti2021optimized,stutz2022learning}. Some methods explicitly aim to improve feature-conditional coverage within
a standard split-conformal prediction framework with marginal guarantees \citep{einbinder2022training,xie2024boosted}.

\subsubsection{Conditioning on the data} \label{sec:calibration-conditional}

The two distinct sources of randomness in the marginal guarantee in Eq.~\ref{eq:coverage} can be explicitly separated using the tower property:
\begin{align}
\P{Y_{n+1}\in C_\alpha(X_{n+1};\mathbf{Z}_{1:n})}
=
\mathbb{E}\!\left[\P{Y_{n+1}\in C_\alpha(X_{n+1};\mathbf{Z}_{1:n}) \mid \mathbf{Z}_{1:n}}\right]. \label{eq:marginal-as-expectation}
\end{align}
This motivates defining the calibration-conditional coverage
\begin{align}
\mathrm{coverage}(\mathbf{Z}_{1:n})
:=
\P{Y_{n+1}\in C_\alpha(X_{n+1};\mathbf{Z}_{1:n}) \mid \mathbf{Z}_{1:n}}, \label{eq:covZ-def}
\end{align}
a random variable.
Marginal validity, Eq.~\ref{eq:coverage}, is equivalent to $\mathbb{E}[\mathrm{coverage}(\mathbf{Z}_{1:n})]\ge 1-\alpha$, which does not say how often unlucky datasets may have $\mathrm{coverage}(\mathbf{Z}_{1:n})$ smaller than $1-\alpha$.

Alternatively, one could demand that $\mathrm{coverage}(\mathbf{Z}_{1:n})$ itself be large with high probability over $\mathbf{Z}_{1:n}$. This matches the traditional notion of a tolerance region \citep{wilks1941determination,scheffe1945non,tukey1947non}, and in modern learning-theoretic terminology  corresponds to a probably-approximately-correct (PAC) coverage guarantee \citep{vovk2012conditional,Park2020PAC}.

\paragraph{PAC coverage}
Fix a confidence parameter $\delta\in(0,1)$. A prediction set $C_{\alpha,\delta}(X_{n+1};\mathbf{Z}_{1:n})$ satisfies PAC coverage at level $(\alpha,\delta)$ if
\begin{align}
\mathbb{P}\left[
\P{Y_{n+1}\in C_{\alpha,\delta}(X_{n+1};\mathbf{Z}_{1:n}) \mid \mathbf{Z}_{1:n}} \ge 1-\alpha
\right]
\ge 1-\delta, \label{eq:pac-coverage}
\end{align}
where the outer probability is taken over $\mathbf{Z}_{1:n}$. In words, with probability at least $1-\delta$ over the observed data $\mathbf{Z}_{1:n}$, the resulting prediction set achieves coverage of at least $1-\alpha$ for the distribution of future test cases.

\paragraph{Split conformal prediction and tolerance regions}
PAC coverage is easiest to understand in the split conformal setting (cf.~Section~\ref{sec:split-conformal}), where the nonconformity score function can be written as $s':\mathcal{X}\times\mathcal{Y}\to\mathbb{R}$.
Given calibration data $\mathbf{Z}_{1:n}=(X_i,Y_i)_{i=1}^n$, define the nonconformity scores $S_i=s'(X_i,Y_i)$ and their order statistics $S_{(1)}\le\cdots\le S_{(n)}$. For any $r\in\{0,1,\ldots,n-1\}$, define
\begin{align}
T_r(X_{n+1}; \mathbf{Z}_{1:n})
&:=\left\{\tilde{y} \in\mathcal{Y}:\#\{i\in[n]: S_i \ge s'(X_{n+1},\tilde{y})\}\ge r+1\right\} \label{eq:Tr-def}\\
&=\left\{\tilde{y} \in\mathcal{Y}: s'(X_{n+1},\tilde{y})\le S_{(n-r)}\right\}. \notag
\end{align}
For a suitable value of $r$, this recovers the conformal prediction set $C_\alpha(X_{n+1}; \mathbf{Z}_{1:n})$ in Eq.~\ref{eq:conf-pred-set}.

If the data are i.i.d.~(and not generally under exchangeability), the calibration-conditional coverage of $T_r(X_{n+1}; \mathbf{Z}_{1:n})$ has an exact beta distribution \citep{wilks1941determination}.

\begin{theorem} \label{thm:beta-conditional-coverage}
Assume that $Z_1,\ldots,Z_{n+1}$ are i.i.d. 
Let $s':\mathcal{X}\times\mathcal{Y}\to\mathbb{R}$ be fixed, and assume the scalar random variable $s'(X,Y)$ has a continuous distribution. Fix $r\in\{0,1,\ldots,n-1\}$ and let $T_r$ be defined by Eq.~\ref{eq:Tr-def}. Then the calibration-conditional miscoverage probability
\begin{align}
\theta_r(\mathbf{Z}_{1:n})
:=
\mathbb{P}\left\{Y_{n+1}\notin T_r(X_{n+1}; \mathbf{Z}_{1:n}) \mid \mathbf{Z}_{1:n}\right\} \label{eq:theta-def}
\end{align}
has distribution $\mathrm{Beta}(r+1,n-r)$.
\end{theorem}

\begin{proof}
Let $F$ denote the CDF of $S:=s'(X,Y)$ under the population distribution. Conditional on $\mathbf{Z}_{1:n}$, $S_{(n-r)}$ is fixed while $s'(X_{n+1},Y_{n+1})$ is an independent draw from the same distribution, hence
$\theta_r(\mathbf{Z}_{1:n})=\mathbb{P}\{S_{n+1}>S_{(n-r)}\mid \mathbf{Z}_{1:n}\}=1-F(S_{(n-r)})$,
where $S_{n+1}:=s'(X_{n+1},Y_{n+1})$.
Now set $W_i:=F(S_i)$. Under the continuity assumption, the probability integral transform implies that $W_1,\ldots,W_n$ are i.i.d.\ $\mathrm{Unif}(0,1)$, and $F(S_{(n-r)})=W_{(n-r)}$ almost surely. Therefore $\theta_r(\mathbf{Z}_{1:n})$ has the same distribution as $1-W_{(n-r)}$.
Now, $W_{(n-r)}\sim \mathrm{Beta}(n-r,r+1)$; see e.g., \cite{arnold2008first}. The result follows.
\end{proof}

Since calibration-conditional coverage has a known distribution that does not depend on the underlying data distribution, following \cite{wilks1941determination} we can choose $r$ to achieve the desired guarantee, including PAC coverage (Eq.~\ref{eq:pac-coverage}) and marginal coverage (Eq.~\ref{eq:coverage}). Clopper--Pearson intervals provide an alternative perspective through the connection between the beta and binomial distributions \citep[e.g.,][]{arnold2008first,Park2020PAC}.

A useful corollary of Theorem~\ref{thm:beta-conditional-coverage} is that $\mathrm{Var}(\theta_r(\mathbf{Z}_{1:n})) \approx \alpha(1-\alpha)/n$ if $r$ is chosen so that $\mathbb{E}[\theta_r(\mathbf{Z}_{1:n})]\approx\alpha$.
For example, when $\alpha=5\%$, this implies fluctuations of roughly $1.5\%$ down to $1\%$ for calibration sizes in the range $n=200$--$500$, motivating the rule of thumb that a few hundred calibration cases are often sufficient in practice.

\paragraph{Beyond i.i.d.: covariate and label shift}

The beta identity assumes that calibration and test scores are i.i.d. Under distribution shift, the inner probability in Eq.~\ref{eq:pac-coverage} must be taken with respect to the test distribution.

In this setting, PAC-style guarantees remain possible in some cases. For covariate shift with weights admitting suitable region-wise confidence intervals, one can apply worst-case rejection sampling to obtain an i.i.d.\ sample from the target distribution \citep{park2021pac}, and then use the standard PAC calibration approach. Asymptotic PAC coverage can also be achieved in a manner that is doubly robust to errors in weight and miscoverage estimation \citep{qiu2023prediction}. Interestingly, predicting missing outcomes under a missing-at-random assumption can be reduced to conformal prediction under covariate shift \citep{lee2025conditional}.

Under label shift, one can construct confidence intervals for label weights, yielding PAC coverage \citep{si2024pac}. This involves confidence intervals for class probabilities and confusion matrix entries,  propagated through matrix inversion \citep{si2024pac}.

\subsection{Conformal Prediction with Weakly Supervised Data}

A key assumption so far has been that the outcome of interest is fully observed in the available data. In many applications, however, outcomes may be only partially or imperfectly observed. In such settings, conformal prediction typically requires additional modeling assumptions and may provide only weaker guarantees. Nonetheless, several extensions of conformal prediction have been developed for these weakly supervised settings.

\paragraph*{Censored time-to-event data.}
In survival analysis, the outcome is an event time often partially censored. \citet{candes2023conformalized} adapt conformal prediction to this setting by focusing on one-sided intervals and recasting the problem as conformal prediction under covariate shift (Section~\ref{sec:weighted-cp}). Subsequent developments reduce conservativeness \citep{gui2024conformalized}, handle more general censoring mechanisms \citep{sesia2025doubly}, and construct two-sided conformal predictions \citep{farina2025doubly,pmlr-v266-sesia25a}. 
These methods rely on standard assumptions such as uninformative censoring and estimated inverse-probability of censoring weights, and therefore provide weaker guarantees than exact finite-sample coverage; for instance asymptotic and doubly robust coverage \citep{yang2024doubly}.

\paragraph*{Individual treatment effects.}
A related challenge arises in causal inference, where prediction of individual treatment effects depends on two counterfactual outcomes that are never jointly observed. This induces a structured distribution shift between observed and target quantities, which can be addressed through appropriate reweighting \citep{lei2021conformal,lee2025conditional} and sensitivity analysis \citep{jin2023sensitivity,yin2024conformal}.

\paragraph*{Proxy or noisy labels.}
Collecting accurate outcomes is sometimes feasible in principle but expensive in practice, prompting the use of surrogate outcomes. Conformal prediction has been extended to such settings \citep{Stutz2023,JMLR:v25:23-0253}. One research line considers data with noisy (occasionally incorrect) labels. \citet{einbinder2024label} show standard conformal methods are often conservative under non-adversarial label noise. 
Several adaptive methods have been developed to leverage models of random label contamination for classification \citep{sesia2024adaptive,clarkson2024split,bortolotti2025noise,penso2025conformal}, regression \citep{cohen2025efficient}, and outlier detection \citep{bashari2025robust}.

\subsection{Further Extensions and Related Methods}
\label{other}

Further extensions of conformal prediction have been proposed. Without aiming for completeness, we briefly highlight several notable directions.

\paragraph*{Batch and selective conformal prediction.}
In many applications, predictions must be made simultaneously for multiple test cases. One challenge is achieving joint validity for the entire batch \citep{lee2024batch,gazin2024powerful}. A second is selective inference, where guarantees are required for data-driven subsets selected from the batch \citep{jin2023selection,bao2024selective,gazin2025selecting,jin2025confidence}. These methods build on classical multiple testing ideas such as the false discovery rate \citep{benjamini1995controlling} and the false coverage rate \citep{benjamini2005false}. Applications include screening eligible patients with long predicted survival for oncology studies \citep{sesia2025distributionfreeselectionlowriskoncology}.

\paragraph*{Aggregating conformal predictions.}
\citet{liang2024conformal}, \citet{liang2022integrative,yang2025selection}, and \citet{liang2023conformal} study how to choose from or aggregate multiple conformal predictions for the same test case obtained using different models. Finite-sample guarantees are retained either by designing aggregation procedures that preserve permutation invariance, or by quantifying how deviations from symmetry affect validity.

\paragraph*{Extended theories.}
\citet{barber2025unifying} show that many conformal methods are special cases of a unified framework, extending the connection to pivots (Section~\ref{sec:foundations-pivots} and \cite{dobriban2023joint}). From this point of view, conformal prediction arises by revealing partial information about the data and deriving a pivotal conditional distribution. For standard split and full conformal methods, what is revealed is the bag of observed values, which induces a distribution that is uniform over permutations. Complementarily, \citet{dobriban2025symmpi} extend conformal prediction to data exhibiting general group symmetries.

\paragraph*{Conformal prediction using $e$-values.}
As an alternative to the $p$-value–based perspective adopted here, conformal prediction can be reformulated using $e$-values \citep{pmlr-v230-balinsky24a,VOVK2025111674}. The concept of $e$-values dates back to \citet{vovk1993empirical} and \citet{10.5555/2074094.2074112}; see \citet{ramdas2025hypothesis} for an overview.
An advantage of $e$-values over $p$-values is that they facilitate aggregation of potentially dependent inferences, whereas $p$-values are generally harder to combine \citep{vovk2020combining,vovk2022admissible}.
In conformal prediction, $e$-values enable sequential prediction with dynamic stopping rules, data-driven selection of coverage levels \citep{gauthier2025values}, the reduction of algorithmic variability \citep{NEURIPS2023_cec8ad77,lee2025full}, and the extension of prediction sets to allow non-binary inclusion \citep{koning2025fuzzy}.

\paragraph*{Alternative frameworks for distribution-free predictive inference.}
More broadly, methods beyond split and full conformal prediction can provide distribution-free predictive inference based on black-box models.  \citet{barber2021predictive} develop jackknife and cross-validation–based approaches. Although introduced for regression, these ideas apply more generally; see e.g., \citet{romano2020classification}. Compared to standard conformal prediction, they are often more data-efficient, at the cost of greater theoretical complexity and somewhat looser guarantees unless models are sufficiently stable \citep{bousquet2002stability}. \citet{kim2020predictive} show these methods integrate naturally with bagging-based learners \citep{breiman1996bagging}, including random forests \citep{breiman2001random}.

\paragraph*{Time series and online prediction.}
When observations exhibit unknown temporal dependence, neither exchangeable nor weighted conformal prediction applies directly, motivating methods for time series and other dependent data streams \citep{xu2021conformal}. Online conformal prediction targets sequential prediction under minimal stochastic assumptions, allowing even deterministic or adversarial sequences.
The goal is to control miscoverage on average over long time horizons, rather than over i.i.d.\ test cases from a population, by updating the nominal level of future conformal prediction sets as new labeled data arrive. Representative methods use online subgradient descent on the quantile loss \citep{gibbs2021adaptive}; multi-valid and grid-based schemes \citep{bastani2022practical}; expert-aggregation avoiding manual step-size tuning \citep{zaffran2022adaptive,gibbs2024conformal}; parameter-free online learning \citep{zhang2024discounted,podkopaev2024adaptive}; and related approaches \citep{bhatnagar2023improved,angelopoulos2023conformalpid,cai2024tractable,srinivas2026online}. Recent work further shows that tools from online learning can be leveraged in this setting, as vanishing linearized regret implies asymptotic coverage \citep{liu2026online}.

\paragraph*{Decision making.}
Another line of research studies how prediction sets can support decision making.
\cite{cresswell2024conformal} show that humans make better data-driven decisions when provided with adaptive conformal prediction sets compared to fixed-size sets with the same coverage guarantee.
\citet{kiyani2025decision} and \citet{wang2026optimal} analyze the optimality of prediction sets for decision-makers minimizing different loss functions.

\begin{summary}[SUMMARY POINTS]
\begin{enumerate}
\item Machine learning has become increasingly attentive to uncertainty and confidence, yet traditional model-based statistical theory may struggle to keep pace with the complexity and rapid evolution of modern data sets and algorithms. Conformal prediction offers a timely and principled statistical response to this challenge.
\item Conformal prediction requires minimal assumptions on data symmetries, such as exchangeability. Guarantees are exact but require thoughtful interpretation.
\item Conformal prediction guarantees marginal coverage regardless of model quality, but informative prediction sets require accurate, uncertainty-aware models.
\item The statistical principles of conformal prediction are simple and flexible, allowing the methodology to be extended in many directions.
\item Conformal methods are well suited to predicting observable quantities; they can be adapted to settings with imperfectly observed outcomes, but are not designed for inference on unobservable population parameters.
\end{enumerate}
\end{summary}

\begin{issues}[FUTURE ISSUES]
\begin{enumerate}
\item With a relatively solid, though still evolving, foundation, the future of conformal prediction may lie in its deeper integration into real-world data science pipelines, AI systems, and data-driven decision making. We therefore cautiously anticipate that many of the most impactful near-term advances will be driven by applications.
\item The practical effectiveness of conformal prediction methods depends on having access to models that can accurately capture predictive uncertainty; conformal prediction therefore complements, rather than substitutes for, continued development in uncertainty-aware machine learning.
\end{enumerate}
\end{issues}


\section*{Acknowledgments}
The authors thank Edgar Dobriban for helpful feedback and contributions to several portions of the paper, especially Section~\ref{sec:calibration-conditional}.
M.S.~was partly supported by a Google Research Scholar award.
S.F.~was partly supported by the Italian Ministry of Education, University and Research (MIUR), ``Dipartimenti di Eccellenza" grant 2023-2027.

\bibliographystyle{abbrvnat}
\bibliography{references}

\clearpage
\appendix
\section{Additional Details on Illustrative Examples}

\subsection{Predicting a Continuous Scalar Variable} \label{app:illustrative-numerical}

This appendix provides additional details related to Section~\ref*{sec:illustrative-numerical}.

\subsubsection{Empirical demonstration}

To demonstrate the performance of conformal prediction intervals, we report in Figure~\ref{fig:example_continuous} the results of numerical experiments based on data simulated from three distinct distributions: a standard normal  $\mathcal{N}(0,1)$; a Student's $t$ with three degrees of freedom; and a normal mixture with three components $(\mathcal{N}(-2,0.01),\mathcal{N}(0,1), \mathcal{N}(2,0.01))$, with weights $(0.09, 0.82, 0.09)$. 
For each setting, we simulate datasets of varying sizes and construct one-sided prediction intervals at level $\alpha = 0.1$.

We compare the conformal method to: \emph{(i)} the population oracle, which serves as the ideal benchmark; \emph{(ii)} an empirical plug-in method that uses the predictive upper bound $Q(\hat{P}(\mathbf{Y}_{1:n});\, 1-\alpha)$, where $\hat{P}(\mathbf{Y}_{1:n})$ denotes the empirical distribution (this method is expected to behave similarly to conformal prediction when $n$ is large);
and \emph{(iii)} a parametric normal prediction interval with upper bound $U_{\alpha}^{\mathrm{norm}}(\mathbf{Y}_{1:n}) = \bar{Y}_{1:n} + t_{1-\alpha,n-1} \cdot \mathrm{sd}(\mathbf{Y}_{1:n})\sqrt{1+n^{-1}}$,
where $\bar{Y}_{1:n} = n^{-1}\sum_{i=1}^{n} Y_i$ and $\mathrm{sd}^2(\mathbf{Y}_{1:n}) = (n-1)^{-1}\sum_{i=1}^{n}(Y_i - \bar{Y}_{1:n})^2$; this procedure is  valid and asymptotically optimal if the population is normal \citep[e.g.,][]{geisser2017predictive,meeker2017statistical}.

Figure~\ref{fig:example_continuous} confirms that conformal prediction intervals converge to the oracle intervals as the sample size $n$ grows, always maintaining marginal coverage at or above $1-\alpha$. When the true data-generating distribution is normal, conformal prediction is only slightly less efficient than the parametric method. However, the latter is only valid under correct model specification; in general, it may substantially over-cover (Student's $t$ example) or under-cover (mixture example). The plug-in approach under-covers in small samples.

\begin{figure}[!htbp]
    \centering
    \includegraphics[width = 0.9\textwidth]{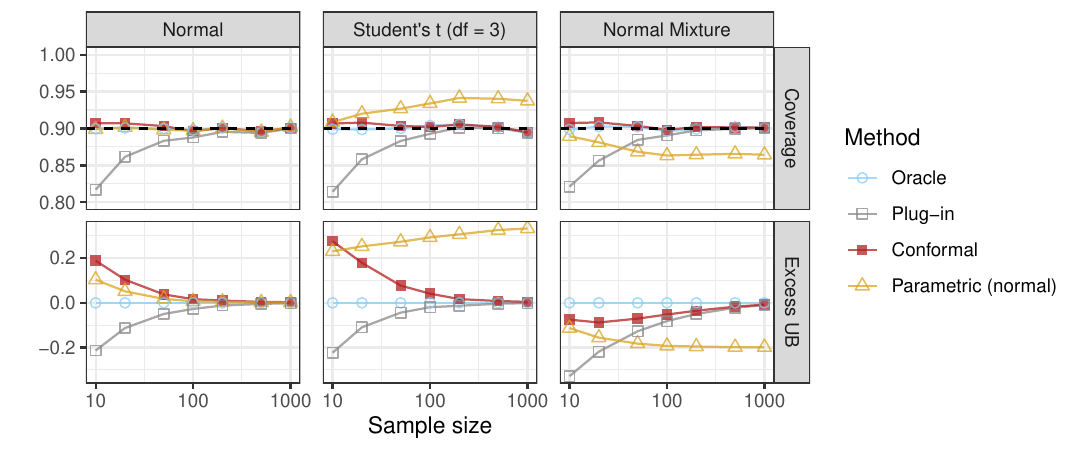}
    \caption{
    Illustrative simulation of one-sided prediction upper bounds for a continuous random variable at level $\alpha = 0.1$, under three different data-generating distributions. Two performance metrics are shown as a function of the sample size: marginal coverage (top) and excess upper bound relative to the ideal population oracle (bottom). The methods compared are conformal prediction, an empirical plug-in approach, a normal parametric prediction interval, and the oracle interval based on the true data-generating distribution. Each curve represents averages over $10{,}000$ independent simulations. Conformal prediction guarantees finite-sample coverage and performs similarly to the oracle as the sample size grows.}
    \label{fig:example_continuous}
\end{figure}

\subsection{Predicting a Categorical Variable} \label{app:illustrative-categorical}

This appendix provides additional details related to Section~\ref*{sec:illustrative-categorical}.

\subsubsection{A randomized oracle}

The oracle prediction sets defined in Section~\ref*{sec:illustrative-categorical} can be made even smaller on average while maintaining marginal coverage through randomization. Using the random features $\xi$, the oracle may decide to exclude the borderline label in some cases, depending on how much the cumulative probability mass exceeds $1-\alpha$. 
Formally, the oracle includes label $y$ in the prediction set if and only if $p^*(y, \xi_{n+1}) := \sum_{k=r(y)+1}^{K} \pi^*_{(k)} + \pi^*_y \cdot \xi_{n+1} > \alpha$, where $\pi^*_{(1)} \ge \pi^*_{(2)} \ge \cdots \ge \pi^*_{(K)}$ are the sorted label probabilities and $r(y)$ is the rank of $\pi^*_y$, so that $\pi^*_y = \pi^*_{(r(y))}$; e.g., see \citet{romano2020classification}.

\subsubsection{Empirical demonstration}

To visualize the performance of the conformal prediction sets described above, 
Figure~\ref{fig:example_categorical} reports the results of numerical experiments based on synthetic data generated from three multinomial distributions over $K=5$ labels. 
Specifically, we consider: 
a \emph{balanced} distribution assigning equal probability $1/5$ to each label; 
a \emph{moderately imbalanced} distribution with probabilities $(0.4,\,0.25,\,0.15,\,0.12,\,0.08)$; 
and a \emph{highly imbalanced} distribution with probabilities $(0.75,\,0.15,\,0.09,\,0.01,\,0)$.
For each setting we generate data sets of varying sizes and construct prediction sets at level $\alpha = 0.1$.

We compare conformal prediction with three benchmarks:
\emph{(i)} the population oracle, which uses knowledge of the true distribution $P^*$; 
\emph{(ii)} an empirical plug-in approach, which replaces $P^*$ by its multinomial maximum-likelihood estimate based on the observed sample $\mathbf{Y}_{1:n}$; and 
\emph{(iii)} a Bayesian approach using a uniform Dirichlet prior.  
In the Bayesian method, we place a Dirichlet$(1,\ldots,1)$ prior on the class probabilities $\pi = (\pi_1,\ldots,\pi_K)$, yielding a Dirichlet$(1+n_1,\ldots,1+n_K)$ posterior after observing label counts $(n_1,\ldots,n_K)$.  
We then apply the oracle construction using the predictive distribution $\mathbb{P}_{\text{Bayes}}[Y_{n+1}=k \mid \mathbf{Y}_{1:n}] = (1 + n_k)/(K + n)$ instead of the true data-generating distribution $P^*$.

\begin{figure}[!htbp]
    \centering
    \includegraphics[width = 0.9\textwidth]{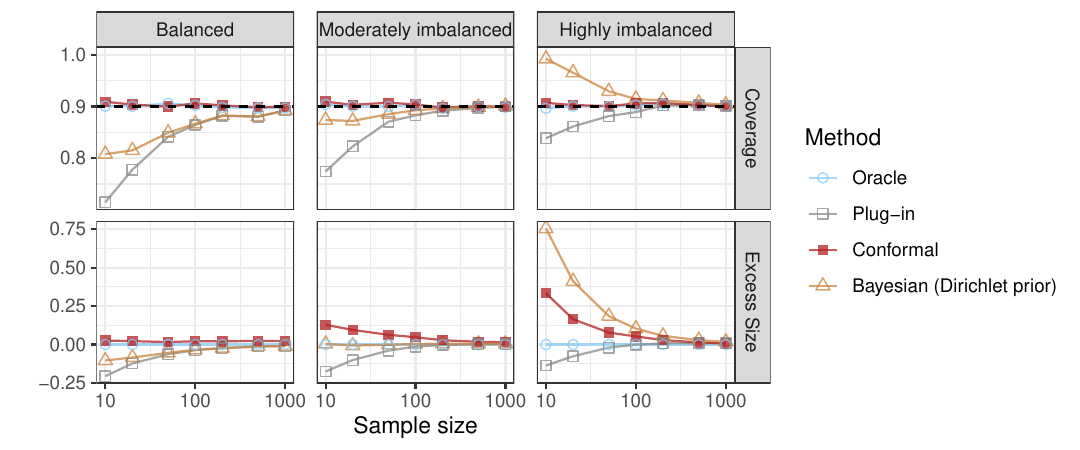}
    \caption{
    Illustration of prediction sets for a categorical outcome at target level $\alpha = 0.1$ under three different data-generating distributions. 
    The top panel shows the marginal coverage as a function of the sample size, and the bottom panel shows the excess size of each method relative to the population oracle.
    The methods compared are conformal prediction, a plug-in estimator, a Bayesian approach with a uniform Dirichlet prior, and the oracle prediction set based on the true data-generating distribution. 
    Results are averaged over $10{,}000$ independent repetitions.
    Conformal prediction maintains finite-sample marginal coverage guarantees and approaches the oracle performance rapidly as the sample size grows.}
    \label{fig:example_categorical}
\end{figure}

Figure~\ref{fig:example_categorical} shows that conformal prediction sets converge quickly to the oracle sets as the sample size increases, while consistently achieving marginal coverage at the desired level $1-\alpha$.  
The plug-in approach, lacking any finite-sample guarantees, under-covers substantially in small samples.  
The Bayesian predictor, although relatively more conservative due to the prior, still does not satisfy finite-sample frequentist coverage guarantees and may either under-cover (balanced and moderately imbalanced cases) or over-cover (highly imbalanced case), depending on how well the prior aligns with the true population distribution.

\FloatBarrier

\section{Diabetes Classification Example using {NHANES} Data} \label{app:diabetes_details}

This appendix describes the data preprocessing, variable construction, model specification, and conformal calibration steps for the diabetes classification example using NHANES (08/2021--08/2023) data (Figure~\ref*{fig:3}), obtained from \citet{paulose2021national}. We merge demographic, examination, laboratory, and questionnaire components using the respondent identifier (SEQN). All preprocessing is performed prior to sample splitting.

\paragraph*{Outcome definition and preprocessing.}
An individual is classified as having diabetes ($Y=1$) if they report a physician diagnosis of diabetes or meet standard laboratory criteria (fasting plasma glucose $\ge 126$ mg/dL or HbA1c $\ge 6.5\%$). Individuals reporting no diagnosis and below both laboratory thresholds are classified as healthy ($Y=0$); observations with missing or indeterminate outcome information are excluded. We restrict the analysis to participants over 30 years of age and exclude pregnant individuals. After additionally removing observations with missing covariates, 2{,}125 individuals remain.

\paragraph*{Risk modeling.}
We fit a logistic regression model including demographic variables (age, sex, race/ethnicity, poverty-income ratio), anthropometric measures (waist circumference and height), cardiometabolic markers (systolic blood pressure, triglyceride-to-HDL ratio, ALT, uric acid, and GGT), and behavioral variables (sleep duration and self-reported physical activity). Continuous variables are used on their natural scales. Age is modeled flexibly using a natural cubic spline with three degrees of freedom, while the remaining covariates enter as linear or categorical main effects.

\paragraph*{Split-conformal methodology.}
The data are randomly partitioned into training (1{,}062), calibration (319), and test (744) sets. The logistic model is fit on the training set to compute an estimate $\hat{\pi}_1(x)$ of $\mathbb{P}(Y=1 \mid X=x)$, which also gives $\hat{\pi}_0(x) = 1 - \hat{\pi}_1(x)$. On the calibration set, we compute probability-based nonconformity scores and determine the empirical $(1-\alpha)(1+1/n)$ quantile with $\alpha=0.05$, yielding a threshold $\hat{\tau} \approx 0.813$. For a new individual, the conformal prediction set is $\widehat{C}(x) = \{y \in \{0,1\}: \hat{\pi}_y(x) \ge 1 - \hat{\tau}\}$, which in the binary setting (with $\hat{\tau} \geq 1/2$) produces one of three possible outputs: $\{\text{Healthy}\}$, $\{\text{Diabetes}\}$, or $\{\text{Healthy},\text{Diabetes}\}$.
Under exchangeability, this guarantees marginal coverage $\mathbb{P}\bigl(Y \in \widehat{C}(X)\bigr) \ge 1-\alpha = 0.95$, where the probability is taken over the joint distribution of $(X,Y)$. Coverage is evaluated on the independent test set as the proportion of individuals whose true label lies in the reported prediction set.


\end{document}

\end{document}